\RequirePackage{fix-cm}
\documentclass[smallcondensed]{svjour3}     % onecolumn (ditto)
\smartqed % flush right qed marks, e.g. at end of proof
\usepackage{graphicx}
\def\makeheadbox{\relax}
\usepackage{
  amsmath,
  amssymb,
  graphicx
}

\usepackage{url}
\usepackage{bbold}
\usepackage{mathtools}

\usepackage[nice]{nicefrac}
\usepackage[table]{xcolor}
\usepackage{algorithm}
\usepackage[noend]{algpseudocode}
\usepackage{tikz}
\usepackage{tkz-berge}
\usetikzlibrary{arrows.meta}
\usetikzlibrary{external,backgrounds}
\usetikzlibrary{positioning}
\usepackage{mathtools}
\usepackage{subcaption}
\usepackage{bbold}
\usepackage{xfrac}
\usepackage{ifthen}
\usetikzlibrary{decorations.markings}
\usepackage{diagbox}
\usepackage{tabularx}
\usepackage{array}
\usepackage[export]{adjustbox}
\usepgfmodule{nonlineartransformations}
\usepackage[shortlabels]{enumitem}
\usepackage{forloop}
\usepackage{ifthen}

\algnewcommand\And{\textbf{and}}
\algnewcommand\Or{\textbf{or}}

\newcolumntype{Y}{>{\centering\arraybackslash}X}

\pgfdeclarelayer{background}
\pgfdeclarelayer{foreground}
\pgfsetlayers{background,main,foreground}

\usetikzlibrary{decorations.pathreplacing,calc}
\tikzset{draw half paths/.style 2 args={%
  decoration={show path construction,
    lineto code={
      \draw [#1] (\tikzinputsegmentfirst) --
         ($(\tikzinputsegmentfirst)!0.2!(\tikzinputsegmentlast)$);
      \draw [#2] ($(\tikzinputsegmentfirst)!0.2!(\tikzinputsegmentlast)$)
        -- (\tikzinputsegmentlast);
    }
  }, decorate
}}

\tikzset{
    my box/.style = {
        , line cap = round
        , line join = round
    }
  }

\newcommand{\highlight}[3]{
  \path [my box, line width = #1, draw = #2, transparency group, opacity=1] #3;
}

\usetikzlibrary{snakes}
\tikzset{wavy/.style={decorate,decoration={snake,amplitude=.4mm,segment length=2mm,post length=0mm,pre length=0mm},line width=.5}}

\tikzset{myDotted/.style={line width=1,dash pattern={on 1pt off 3pt}}}
\tikzset{myDashed/.style={line width=1,dash pattern={on 3pt off 3pt}}}
\tikzset{myDashDotted/.style={line width=1,dash pattern={on 3pt off 3pt on 1pt off 3pt}}}

\SetVertexNormal[FillColor=gray!25]

\makeatletter
\renewcommand*{\write@math}[3]{%
            \pgfmathtruncatemacro{\printindex}{#3+1}
            \Vertex[x = #1,y = #2,%
                    L = \cmdGR@cl@prefix\grMathSep{\printindex}]{\cmdGR@cl@prefix#3}}

\def\R{\mathbb R}

\def\Z{\mathbb Z}
\def\N{\mathbb N}
\def\E{\mathbb E}
\def\1{\mathbb 1}

\def\mc{\mathcal}
\def\symmgroup{\mathfrak{S}}
\def\cyclicSymmgroup{\symmgroup_{C}}

\DeclarePairedDelimiter\ceil{\lceil}{\rceil}
\DeclarePairedDelimiter\floor{\lfloor}{\rfloor}

\tikzset{v 0/.style={fill=white}, v 1/.style={fill=black!40},
   venn path 1/.style={insert path={
     (90:1/sqrt 3) arc (60:120:1) arc (180:0:1) arc (60:120:1) -- cycle }},
   venn path 2/.style={insert path={
     (90:1/sqrt 3) arc (120:180:1) arc (240:180:1) arc (120:60:1) -- cycle }},
   venn path 3/.style={insert path={
     (90:1/sqrt 3) arc (120:180:1) arc (240:300:1) arc (0:60:1) -- cycle }},
   pics/venn 3/.style args={#1#2#3#4#5#6#7#8}{code={%
     \fill [v #1] (-1.5-.5, -1.3-.5) rectangle (1.5+.5,1.6+.5);
     \fill [v #2, rotate=240, venn path 1];
     \fill [v #3, rotate=120, venn path 1];
     \fill [v #4, venn path 1];
     \fill [v #5, rotate=240, venn path 2];
     \fill [v #6, rotate=120, venn path 2];
     \fill [v #7, venn path 2];
     \fill [v #8, venn path 3];
     \draw [line width = .05mm] (90:1/sqrt 3) circle [radius=1] (210:1/sqrt 3) circle [radius=1]
     (330:1/sqrt 3) circle [radius=1] (-1.5-.5, -1.3-.5) rectangle (1.5+.5,1.6+.5);
     \node at (0,1.65){\Large $E_1$};
     \begin{scope}[rotate=120]
       \node at (0,1.65){\Large $E_2$};
     \end{scope}
     \begin{scope}[rotate=-120]
       \node at (0,1.65){\Large $E_3$};
     \end{scope}
}}}
\pgfmathsetbasenumberlength{8}% Still very important!

\tikzset{every venn/.style={x=1em, y=1em, baseline=-.666ex,
    v 1/.style={fill=gray}}}

\makeatletter
\def\mytransformation{%
\pgfmathsetmacro{\myX}{\pgf@x}
\pgfmathsetmacro{\myY}{\pgf@y}
\setlength{\pgf@x}{\myX pt}
\setlength{\pgf@y}{\myY pt}
}
\makeatother

\spnewtheorem{numberedClaim}{Claim}{\bfseries}{\itshape}

\begin{document}
\title{Matchings under distance constraints II.}
\author{P\'eter Madarasi}

\institute{P. Madarasi \at
  Department of Operations Research, ELTE E\"otv\"os Lor\'and University and the HUN-REN–ELTE Egerváry Research Group on Combinatorial Optimization, P\'azm\'any P\'eter s\'et\'any 1/C, 1117 Budapest, Hungary. \email{madarasip@staff.elte.hu}
}

\date{}

\maketitle

\thispagestyle{empty}

\begin{abstract}
  This paper introduces the \emph{$d$-distance $b$-matching problem}, in which we are given a bipartite graph $G=(S,T;E)$ with $S=\{s_1,\dots,s_n\}$, a weight function on the edges, an integer $d\in\Z_+$ and a degree bound function $b:S\cup T\to\Z_+$.
  The goal is to find a maximum-weight subset $M\subseteq E$ of the edges satisfying the following two conditions: 1) the degree of each node $v\in S\cup T$ is at most $b(v)$ in $M$, 2) if $s_it,s_jt\in M$, then $|i-j|\geq d$.
  In the cyclic version of the problem, the nodes in $S$ are considered to be in cyclic order.
  We get back the \emph{(cyclic) $d$-distance matching problem} when $b(s) = 1$ for $s\in S$ and $b(t) = \infty$ for $t\in T$.

  We prove that the $d$-distance matching problem is APX-hard, even in the unweighted case.
  We show that $(2-\frac{1}{d})$ is a tight upper bound on the integrality gap of the natural integer programming model for the cyclic $d$-distance $b$-matching problem provided that $(2d-1)$ divides the size of $S$.
  For the non-cyclic case, the integrality gap is shown to be at most $(2-\frac{2}{d})$.
  The proofs give approximation algorithms with guarantees matching these bounds, and also improve the best known algorithms for the (cyclic) $d$-distance matching problem.

  In a related problem, our goal is to find a permutation of $S$ maximizing the weight of the optimal $d$-distance $b$-matching.
  This problem can be solved in polynomial time for the (cyclic) $d$-distance matching problem --- even though the (cyclic) $d$-distance matching problem itself is NP-hard and also hard to approximate arbitrarily.
  For (cyclic) $d$-distance $b$-matchings, however, we prove that finding the best permutation is NP-hard, even if $b\equiv 2$ or $d=2$, and we give $e$-approximation algorithms.

\end{abstract}

\newpage

\section{Introduction}\label{sec:distance2:intro}
In this paper, we introduce a natural generalization of the \emph{$d$-distance matching problem}~\cite{madarasi2021matchings}, where the degree upper bound function in $S$ can be other than the all-one function, and degree bounds can be posed on the nodes in $T$ as well.
Given a bipartite graph $G=(S,T;E)$ with $S=\{s_1,\ldots, s_n\}$, a positive integer $d\in\N$ and a function $b: S\cup T \to \Z_+$, an edge set $M \subseteq E$ is called \emph{$d$-distance $b$-matching} if 1) the degree of each node $v \in S\cup T$ is at most $b(v)$ in $M$ and 2) if $s_it,s_jt \in M$ for $i\neq j$, then $|i-j|\geq d$.
A $d$-distance $b$-matching is called \emph{perfect} if the degree of each node $s\in S$ is exactly $b(s)$.
In the \emph{cyclic} version of the problem, the nodes in $S$ are considered to be in a cyclic order, and 2) is required cyclically, that is, if $s_it,s_jt\in M$ for $i\neq j$, then both $|i-j|\geq d$ and $|i-j|\leq n-d$ must hold.

In the (cyclic) $d$-distance $b$-matching problem, the goal is to find a maximum-weight (cyclic) $d$-distance $b$-matching for a given weight function $w: E \to \R$.
The special case $w\equiv 1$ is referred to as the \emph{unweighted} problem.
Note that the special case when $b(v) = 1$ if $v\in S$ and $\infty$ if $v\in T$ is the $d$-distance matching problem, which was introduced in an earlier article from the same author~\cite{madarasi2021matchings}.

The (perfect) $d$-distance $b$-matching problem is not only a natural problem extending the literature of matchings, but it also appears in several applications, which are natural extensions of the situations in which the $d$-distance matching problem could be applied, see~\cite{madarasi2021matchings}.
For example, imagine $n$ consecutive all-day events $s_1,\dots,s_n$, each of which must be assigned $b(s_i)$ of watchmen $t_1,\dots,t_k$.
For each event $s_i$, a set of possible watchmen is given --- those who are qualified to be on guard at event $s_i$.
Appoint exactly $b(s_i)$ watchmen to event $s_i$ such that no watchman is assigned to more than one of any $d$ consecutive events, where $d\in\N$ is given, and each watchman $t_j$ is on guard at most $b(t_j)$ events.
In the \emph{weighted} version of the problem, let $w_{s_it_j}$ denote the level of safety of event $s_i$ if watchman $t_j$ is on watch, and the objective is to maximize the level of overall safety.

As another application of the above question, consider $n$ items $s_1,\dots,s_n$ one after another on a conveyor belt, and $k$ machines $t_1,\dots,t_k$.
Each item $s_i$ is to be processed on the conveyor belt by $b(s_i)$ of the qualified machines ${N(s_i)\subseteq\{t_1,\dots,t_n\}}$ such that if a machine processes item $s_i$, then it cannot process any of the next $(d-1)$ items --- because the conveyor belt is running.

\paragraph{\normalfont\textbf{Previous work}}
In the special case $d=|S|$, one gets the classic $b$-matching problem in bipartite graphs.
For $d=1$, the problem reduces to the \mbox{$b$-matching problem} in bipartite graphs, and we will see that it is a special case of the circulation problem for $d=2$.

A feasible $d$-distance $b$-matching $M$ can be thought of as a $b$-matching that does not contain the edge set $\{s_it,s_jt\}$ for any $t\in T$ and $|i-j|\leq d$.
A similar problem is the \emph{$K_{p,p}$-free $p$-matching problem}~\cite{Makai07}.
Here one is given an arbitrary family $\mc T$ of the subgraphs of $G$ isomorphic with $K_{p,p}$, and the goal is to find a maximum-cardinality $b$-matching which does not induce any subgraph in $\mc T$, where $b:S\cup T\to\{0,\dots,p\}$.
This problem can be solved in polynomial time.
Note that in the $d$-distance $b$-matching problem, $b$ is arbitrary and the type of the forbidden subgraphs is $K_{2,1}$.
Another similar problem is the following.
Given a partition $E_1,\dots,E_k$ of $E$ and positive integers $r_1,\dots,r_k$, find a perfect matching $M$ for which $|M\cap E_i|\leq r_i$.
The problem is introduced and shown to be NP-complete in~\cite{Itai78}.
Note that the side constraints in the distance matching problem are similar, but the degree constraints are different and our edge sets do not form a partition of $E$.
Several other versions of the ``restricted'' ($b$-)matching problem have been introduced, for example in~\cite{baste2019uniquelyRestrMatch,BercziVegh10,furst2019acyclicMatching,pap2005alternating}.

The perfect $d$-distance matching problem is a special case of the list-coloring problem on interval graphs~\cite{ZeithoferWess03} and also of the frequency assignment problem~\cite{Aardal2007}, as it was shown in~\cite{madarasi2021matchings}.

% based on the Contribution from part I.
The $d$-distance matching problem was shown to be NP-hard and an FPT algorithm parameterized by $d$ was given in~\cite{madarasi2021matchings}.
An efficient algorithm was also described for the case when the size of $T$ is a constant.
A $(2-\frac{1}{2d-1})$-approximation algorithm for the weighted $d$-distance matching problem was given, which also implies that the integrality gap of the natural IP model is at most this value.
We also gave a $(3/2+\epsilon)$-approximation algorithm for any constant $\epsilon>0$ in the unweighted case.

\paragraph{\normalfont\textbf{Our results}}
We investigate the integrality gap of the natural IP model and give approximation algorithms for the (cyclic) $d$-distance $b$-matching problem in Section~\ref{sec:distance2:approxAndGap}.
In particular, we show that $(2-\frac{1}{d})$ is a tight upper bound on the integrality gap in the cyclic case provided that $(2d-1)$ divides the size of $S$.
Concerning the non-cyclic case, the integrality gap is shown to be at most $(2-\frac{2}{d})$ for $d\geq 2$.
In addition, the proofs provide approximation algorithms with approximation factors matching the bounds above, further improving the algorithms for the (cyclic) $d$-distance matching problem given in an earlier article~\cite{madarasi2021matchings}.
As a special case, this further improves the bound on the integrality gap and the approximation factor for the $d$-distance matching problem from $(2-\frac{1}{2d-1})$ to $(2-\frac{2}{d})$.

Answering an open question from~\cite{madarasi2021matchings}, Section~\ref{sec:distance2:hardness} proves that the (cyclic) $d$-distance matching problem is APX-hard, even in the unweighted case.

In Section~\ref{sec:distance2:optimalPermutations}, motivated by the second application mentioned above, a different aspect of the problem is considered, in which our goal is to find a permutation of $S$ maximizing the weight of the optimal $d$-distance $b$-matching.
We prove that a permutation of $S$ maximizing the weight of the optimal $d$-distance matching can be found in polynomial time --- even though the (cyclic) $d$-distance matching problem itself is NP-hard and also hard to approximate arbitrarily.
For (cyclic) $d$-distance $b$-matchings, however, we prove that finding the best permutation is NP-hard, even when $b(s)=2$ for all $s\in S$ or $d=2$, and we give $e$-approximation algorithms for both the cyclic and the non-cyclic cases.

\paragraph{\normalfont\textbf{Notation}}
Throughout the paper, assume that $G=(S,T;E)$ is a bipartite graph without loops or parallel edges, unless stated otherwise.
Let $\Delta(v)$ and $N(v)$ denote the set of edges incident to node $v$ and the neighbors of $v$, respectively.
For a subset $X\subseteq E$ of the edges, $N_X(v)$ denotes the neighbors of $v$ for edge set $X$.
We use $\deg(v)$ to denote the degree of node $v$.
Let $L_d(s_i)$ and $R_d(s_i)$ denote the nodes in the interval of length (at most) $d$ ending and starting at $s_i$, respectively, that is, $L_d(s_i)=\{s_{\max(i-d+1,1)},\dots,s_i\}$ and $R_d(s_i)=\{s_i,\dots,s_{\min(i+d-1,|S|)}\}$.
In the cyclic case, $L_d(s_i)=\{s_{i-d+1},\dots,s_i\}$ and $R_d(s_i)=\{s_i,\dots,s_{i+d-1}\}$, where the indices are taken modulo $|S|$.
By definition, the minimum and the maximum of the empty set are $\infty$ and $-\infty$, respectively.
Given a function $f:A\to B$, both $f(a)$ and $f_a$ denote the value $f$ assigns to $a\in A$, and let $f(X)=\sum_{a\in X}f(a)$ for $X\subseteq A$.
Let $\chi_Z$ denote the characteristic vector of set $Z$, that is, $\chi_Z(y)=1$ if $y\in Z$, and $0$ otherwise.
Occasionally, the braces around sets consisting of a single element are omitted, for example $\chi_e=\chi_{\{e\}}$ for $e\in E$.
Let $\N$ and $\Z_+$ denote the set of positive and non-negative integers, respectively.

\section{Integrality gap and approximation algorithms}\label{sec:distance2:approxAndGap}
In this section, we prove that $(2-\frac{1}{d})$ is a tight upper bound on the integrality gap of the natural IP model of the $d$-distance $b$-matching problem provided that the size of $S$ is divisible by $(2d-1)$.
Then, we show that $(2-\frac{2}{d})$ is an upper bound on the integrality gap of the non-cyclic version for $d\geq 2$ --- without any restriction on the size of $S$.
The proofs also give two approximation algorithms with approximation factors matching the bounds above.

Consider the following LP-relaxation of the natural IP model for the weighted $d$-distance $b$-matching problem.
\begin{subequations}
  \begin{align}\label{lp:distance2:dmLp}
    \tag{LP1}
    \max&\sum_{e \in  E} w_{e} x_{e}\\
    \mbox{s.t.}\quad\quad\quad\quad\quad&&\nonumber\\
    x&\in\R_+^{E}&\\
    \sum_{e \in \Delta(v)} x_{e} &\leq b(v) &\forall v\in S\cup T\label{lp:distance2:dmLp:eq:degree}\\
    \sum_{\substack{st \in \Delta(t) \\ s \in R_d(s_i)}} x_{st} &\leq 1&\forall t \in T\ \forall i \in \{1,\ldots, n-d\}.\label{lp:distance2:dmLp:eq:distanceConstr}
  \end{align}
\end{subequations}
The LP model for the cyclic case consists of the same conditions but~(\ref{lp:distance2:dmLp:eq:distanceConstr}) is required for all $t \in T$ and for all $i \in \{1,\ldots, n\}$.
This model will be denoted by \textit{(LP1')}.
When $b(s)=1$ for all $s\in S$ and $b(t)=\infty$ for all $t\in T$, the integer solutions to~(\ref{lp:distance2:dmLp}) correspond to the feasible $d$-distance matchings, and we get back the linear program investigated in~\cite{madarasi2021matchings}.

The following theorem gives an upper bound on the integrality gap for the cyclic $d$-distance $b$-matching problem.

\begin{theorem}\label{thm:distance2:gapCyclicDistance}
  If $(2d-1)$ divides $|S|$, then the integrality gap of~{(LP1')} for the weighted cyclic $d$-distance $b$-matching problem is at most $(2-\frac{1}{d})$, and this bound is tight.
  Furthermore, there exists a polynomial-time approximation algorithm with the same guarantee.
\end{theorem}
\begin{proof}
  For every $i\in\{1,\ldots,2d-1\}$, let $S_i \subseteq S$ denote the union of the sets $R_d(s_{i+q(2d-1)})$ for $q\in\{0, \ldots, \frac{n}{2d-1}-1\}$.
  Since the size of $S$ is divisible by $(2d-1)$, the nodes in $S_i$ form intervals of length $d$ in $s_1,\dots,s_n$ and each of these intervals is followed by $(d-1)$ nodes of $S \setminus S_i$ cyclically.
  For each $i\in\{1,\ldots,2d-1\}$, let $G_i=(S,T;E_i)$ be the subgraph of $G$ on the node set of $G$ whose edge set $E_i$ consists of the edges induced by $S_i$ and $T$.

  First, we prove that the polytope given by~{(LP1')} for $G_i$ is the convex hull of its integer solutions.
  Observe that constraints~(\ref{lp:distance2:dmLp:eq:distanceConstr}) need to be required only for those intervals that are fully included in $S_i$, because these immediately imply that the constraints hold for the rest of the intervals.
  The matrix of this reduced linear program is the transpose of the incidence matrices of two laminar families written under each other, which is a well-known network matrix, and the right-hand side is integer, hence the polytope is integer~\cite[Page 152]{AF11}.

  Next, we prove the bound on the integrality gap.
  Let $M_i$ denote a maximum-weight cyclic $d$-distance $b$-matching in $G_i$, and let $M$ be a maximum-weight solution among $M_1, \ldots , M_{2d-1}$.
  Let $x \in \R_+^{E}$ be an optimal LP solution for $G$ and let $M^*$ be a maximum-weight $d$-distance $b$-matching.
  It is easy to see that all edges of $G$ appear in exactly $d$ of the graphs $G_1, \ldots , G_{2d-1}$, which means that
  \[
  wx = \sum_{e \in E} w_ex_e = \frac{1}{d} \sum_{i=1}^{2d-1} \sum_{e \in E_i} w_ex_e
  \]
  holds.
  Restricting an optimal LP solution for $G$ to the edge set of $G_i$, a feasible LP solution is obtained for $G_i$, so the objective value of this restricted solution can be bounded from above by the LP optimum for $G_i$, which is equal to the IP optimum $w(M_i)$.
  From these, one gets that
  \begin{equation}\label{eq:distance2:gapCyclicDistance:mainineq}
    wx \leq \frac{1}{d} \sum_{i=1}^{2d-1} w(M_i) \leq \frac{2d-1}{d} w(M) \leq \frac{2d-1}{d} w(M^*),
  \end{equation}
  since $M$ is a feasible integer solution for $G$.
  This means that the integrality gap is at most $(2-\frac{1}{d})$, which was to be proven.

  Next, we give a tight example for every $d\in\N$.
  Let $G=(S,T;E)$ be a complete bipartite graph, where $S=\{s_1, \ldots, s_{2d-1}\}$ and $T=\{t\}$.
  Let $b(s)=1$ for all $s\in S$, and let $b(t)=\infty$.
  For $w\equiv 1$, the IP optimum is $1$, and $x \equiv \frac{1}{d}$ is an optimal LP solution, meaning that the LP optimum is $(2-\frac{1}{d})$, hence the bound above is tight.

  In fact, this proof shows that $M$ is a $(2-\frac{1}{d})$-approximate solution, which can be found in polynomial time by solving~{(LP1')}~\cite{tardos85CombinatorialLP} for every graph $G_i$, therefore we also obtain an approximation algorithm for those instances of the maximum-weight cyclic $d$-distance $b$-matching problem in which $(2d-1)|n$, which completes the proof.
  \qed
\end{proof}

The next theorem improves this upper bound in the non-cyclic case when $b(t)=\infty$ for all $t\in T$.
As a special case, this also improves the best known upper bound on the integrality gap for the non-cyclic $d$-distance matching problem from $(2-\frac{1}{2d-1})$~\cite{madarasi2021matchings} to $(2-\frac{2}{d})$.

\begin{theorem}\label{thm:distance2:gapDistance}
  Let $b:S\cup T\to\Z_+$ be such that $b(t)=\infty$ for all $t\in T$, and let $d\geq 2$.
  The integrality gap of (\ref{lp:distance2:dmLp}) for the weighted $d$-distance $b$-matching problem is at most $(2-\frac{2}{d})$.
  Furthermore, there exists a polynomial-time approximation algorithm with the same guarantee.
\end{theorem}
\begin{proof}
  Let $G=(S,T;E)$ be a bipartite graph, where $S=\{s_1,\ldots, s_n\}$, and let $d\in\N$ and $w: E \to \R_+$.
  If $(2d-2) \nmid n$, then add $(2d-2-r)$ new isolated nodes to the end of $S$, where $r$ is such that $0<r<2d-2$ and $n=k(2d-2) + r$ for some $k\in\Z_+$.
  This leaves the feasible $d$-distance $b$-matchings in $G$ unchanged, therefore one can assume without loss of generality that $(2d-2)|n$.

  We proceed similarly to the first part of the proof of Theorem~\ref{thm:distance2:gapCyclicDistance}, but now we leave out only $(d-2)$ consecutive nodes --- instead of $(d-1)$ --- between disjoint intervals of length $d$ of $S$.
  That is, let $S_i \subseteq S$ denote the union of the sets $R_d(s_{i+q(2d-2)})$ for $q\in\{0,\ldots, \frac{n}{2d-2}-1\}$, where $i\in\{1, \ldots, 2d-2\}$ and $R_d$ is to be taken in the cyclic sense.
  Let $E_i$ consist of the edges induced by $S_i$ and $T$, and let $G_i=(S,T;E_i)$ for $i\in\{1, \ldots, 2d-2\}$.
  Just as in the proof of Theorem~\ref{thm:distance2:gapCyclicDistance}, we prove that the polytope defined by~(\ref{lp:distance2:dmLp}) is integer for $G_i$.
  \begin{numberedClaim}\label{ncl:distance2:gap}
    For each $i\in\{1,\dots,2d-2\}$, the polytope defined by~(\ref{lp:distance2:dmLp}) is integer for the subgraph $G_i$.
  \end{numberedClaim}
  \begin{proof}
    Without loss of generality, we can assume that $G$ is a complete bipartite graph, and hence the edge set of $G_i$ is the complete graph between $S_i$ and $T$.
    Notice that~(\ref{lp:distance2:dmLp:eq:degree}) is required only for $s\in S_i$, and~(\ref{lp:distance2:dmLp:eq:distanceConstr}) only for the intervals of $S$ contained in $S_i$ and for the intervals of length $d$ containing the last and the first nodes in two consecutive intervals in $S_i$, since the rest of the constraints are redundant.
    It suffices to prove that the matrix of this reduced~(\ref{lp:distance2:dmLp}) is a network matrix~\cite[Page 152]{AF11}.
    Fix an arbitrary order $t_1, \ldots , t_k$ of the nodes in $T$.
    The columns of the matrix of the program, that is the variables, are ordered as follows.
    Let the columns corresponding to the edges incident to $t_j$ form an interval for all $j\in\{1,\dots,k\}$, which appear in the order given by $t_1,\dots,t_k$, and for each $t_j$, sort the interval of the edges incident to $t_j$ by the index of their endpoint in $S$.
    Assume that the rows corresponding to constraints~(\ref{lp:distance2:dmLp:eq:distanceConstr}) appear first, in lexicographical order, and then the rows corresponding to constraints~(\ref{lp:distance2:dmLp:eq:degree}) follow, also in lexicographical order.
    Let $L$ denote the matrix obtained this way, and let $B_t$ denote the submatrix of $L$ given by the edges incident to $t$ and by the constraints~(\ref{lp:distance2:dmLp:eq:distanceConstr}) for $t$.
    By construction,
    \[
    B_t=\left[
      \arraycolsep=1.4pt
      \begin{array}{*{20}c}
        1&\ldots&1&&&&&&&&&&\\
         &&1&1&&&&&&&&&\\
         &&&1&\ldots&1&&&&&&&\\
         &&&&&1&1&&&&&&\\
         &&&&&&&\ddots&&&\\
         &&&&&&&&1&1&&&\\
         &&&&&&&&&1&\ldots&1
      \end{array}
    \right],
    \]
    where the zero entries are omitted, and each row contains either two or $d$ one entries.
    The lines correspond to constraints~(\ref{lp:distance2:dmLp:eq:distanceConstr}) alternately for an interval of $S_i$ and for the interval containing the last and the first node of two consecutive intervals in $S_i$.
    For every $t \in T$, there is one such block $B_t$ in $L$ placed diagonally.
    Furthermore, the rows corresponding to constraints~(\ref{lp:distance2:dmLp:eq:degree}) give $k$ identity matrices side by side, one under each $B_t$, that is, $L$ looks as follows.
    \[
    L=
    \begin{bmatrix}
      \begin{bmatrix}B_{t_1}\end{bmatrix}&&&\\
      &\begin{bmatrix}B_{t_2}\end{bmatrix}&&\\
      &&\ddots &\\
      &&&\begin{bmatrix}B_{t_k}\end{bmatrix}\\
      \\
      \begin{bmatrix}~I~\end{bmatrix}&\begin{bmatrix}~I~\end{bmatrix}& \ldots &\begin{bmatrix}~I~\end{bmatrix}
    \end{bmatrix}
    \]
    Now, we prove that $L$ is a network matrix.
    Note that each column of $L$ contains either two or three ones.
    First, consider the submatrix $L'$ formed by the columns of $L$ containing exactly three ones --- we will handle the rest of the columns later.
    Deleting the full-zero rows from the matrix $L'$, which were created by deleting some of the columns from $L$, we get that
    \[
    L'=
    \begin{bmatrix}
      \begin{bmatrix}A_{t_1}\end{bmatrix}&&&\\
      &\begin{bmatrix}A_{t_2}\end{bmatrix}&&\\
      &&\ddots &\\
      &&&\begin{bmatrix}A_{t_k}\end{bmatrix}\\
      \\
      \begin{bmatrix}~I~\end{bmatrix}&\begin{bmatrix}~I~\end{bmatrix}& \ldots &\begin{bmatrix}~I~\end{bmatrix}
    \end{bmatrix}\text{, where }
    A_t=
    \arraycolsep=1.4pt
    \begin{bmatrix}
      1&&&&&\\
      1&1&&&&\\
      &1&1&&&\\
      &&&\ddots&&\\
      &&&&1&1\\
      &&&&&1
    \end{bmatrix}
    \]
    for each $t \in T$.

    We prove that $L'$ is a network matrix.
    Denote the size of the identity matrices in the last rows of $L'$ by $m$.
    Then each block $A_t$ consists of $(m+1)$ rows and $m$ columns.
    Let $M'$ denote the submatrix given by the last $m$ rows of $L'$, that is, the identity matrices.
    Let the tree $F$ be a path $P$ with $m$ (undirected) edges $f_1, \ldots, f_{m}$, which will correspond to the rows of $M'$.
    The orientation of these edges will be given later.
    For each node of the path $P$, add $k$ new leaves connected to that node.
    Let $e_1^i,  \ldots, e_{k}^i$ denote those newly added leaf edges which are incident to the $i^{\text{th}}$ node of $P$ for $i\in\{1,\ldots,m+1\}$.
    Let edge $e_j^i$ correspond to the $((m+1)(j-1) + i)^{\text{th}}$ row of $L'$, in other words, $e_j^i$ belongs to the row containing the $i^{\text{th}}$ row of $A_{t_j}$.
    Orient the edges of $P$ alternately along the path --- the first edge can be oriented arbitrarily, and this determines the direction of the other edges along the path.
    If the (at most) two arcs of $P$ adjacent to arc $e_i^j$ are oriented towards $e_i^j$, then we orient $e_i^j$ outwards from the common node.
    Otherwise, if the (at most) two arcs of $P$ are oriented away from $e_i^j$, then $e_i^j$ is oriented inwards --- since the arcs of $P$ are alternately oriented, only these two cases are possible.

    Now, we define the non-tree arcs, which correspond to the columns of the matrix.
    Each column intersects exactly one of the matrices in the diagonal, say $A_{t_j}$, and each column contains exactly three non-zero elements.
    Suppose that the $r^{\text{th}}$ column is the $i^{\text{th}}$ column of $A_{t_j}$ for some $j\in\{1,\ldots,k\}$.
    Then two of the three ones in the column are in the $i^{\text{th}}$ and $(i+1)^{\text{st}}$ rows of $A_{t_j}$, to which rows the corresponding arcs are $e_j^{i}$ and $e_j^{i+1}$, respectively.
    The third one is in the  $j^{\text{th}}$ identity matrix, in the row corresponding to arc $f_i$.
    In the tree $F$, arcs $e_j^{i}$, $f_i$ and $e_j^{i+1}$  form a directed path, hence one can add a non-tree arc from the target of this path to its source, which corresponds to the $r^{\text{th}}$ column.
    This shows that $L'$ is a network matrix.

    Next, we prove that the original $L$ is also a network matrix by a simple extension of the tree $F$ and the non-tree arcs defined above.
    Let $M$ denote the submatrix of the last $|E_i|$ rows of $L$, that is, $M$ consists of the identity matrices of size $|E_i| \times |E_i|$ written side by side.

    Observe that if there is exactly one non-zero element in the $r^{\text{th}}$ column of $B_{t_j}$ in $L$ for some $j$, then there is exactly one non-zero element in the $r^{\text{th}}$ column of every block $B_{t_{j'}}$ for $j'\in\{1,\dots,k\}$.
    Each of these columns in $L$ contains exactly two ones.
    The first one is in the $i^{\text{th}}$ row of the corresponding block $B_{t_j}$, which is associated with arc $e_j^i\in F$.
    The other one is in the $r^{\text{th}}$ row of $M$.
    For each $j\in\{1,\ldots, k\}$, add a new leaf to $F$ connected to the endpoint of $e_j^i$ belonging to $P$, and associate it with the $r^{\text{th}}$ row of $M$.
    Orient it such that it forms a directed path of length two with $e_j^i$.
    Finally, add a non-tree arc from the source of this path to its target, which corresponds to the column of $L$ containing the $r^{\text{th}}$ column of $B_{t_j}$.
    This shows that $L$ is indeed a network matrix.
    As the right-hand side of~(\ref{lp:distance2:dmLp}) is integer, we get that the polytope defined by~(\ref{lp:distance2:dmLp}) is integer for $G_i$, which completes the proof of the lemma.
    \qed
  \end{proof}

  We continue the proof of Theorem~\ref{thm:distance2:gapDistance}.
  Let $x \in \R_+^{E}$ be an optimal LP solution, and let $x^{(i)} \in \R_+^{E}$ be an optimal LP solution for $G_i$, where $i\in\{1, \ldots, 2d-2\}$.
  Let $M^*$ be a maximum-weight $d$-distance $b$-matching, and let $M_i$ be a maximum-weight $d$-distance $b$-matching in $G_i$, where $i\in\{1, \ldots, 2d-2\}$.
  Chose a maximum-weight solution among $M_1, \ldots , M_{2d-2}$, and denote it by $M$.

  Similarly to~(\ref{eq:distance2:gapCyclicDistance:mainineq}), as each edge of $G$ is contained in exactly $d$ of the graphs $G_1,  \ldots, G_{2d-2}$, and by Claim~\ref{ncl:distance2:gap}, we get that
  \begin{equation*}
    wx
    \leq \frac{1}{d} \sum_{i=1}^{2d-2} \sum_{e \in E_i} w_ex^{(i)}_e
    =\frac{1}{d} \sum_{i=1}^{2d-2} w(M_i)
    \leq \frac{2d-2}{d}  w(M)
    \leq \frac{2d-2}{d}  w(M^*),
  \end{equation*}
  which means that the integrality gap is at most $(2-\frac{2}{d})$.
  The proof also shows that the edge set $M$ is a $(2-\frac{2}{d})$-approximate solution, which can be found in polynomial time, so the proof also gives an approximation algorithm for the maximum-weight $d$-distance $b$-matching problem with the same guarantee, provided that $b(t)=\infty$ for all $t\in T$.
  \qed
\end{proof}

Observe that, for $d=2$, Claim~\ref{ncl:distance2:gap} holds for arbitrary $b$, therefore~(\ref{lp:distance2:dmLp}) describes the $d$-distance $b$-matching polytope.

\begin{remark}
  In the proofs of Theorems~\ref{thm:distance2:gapCyclicDistance}~and~\ref{thm:distance2:gapDistance}, we constructed a collection of edge sets such that the linear program becomes integer when the problem is restricted to any of them, and every edge is contained in $d$ of the selected edge sets.
  This means that these two collections of edge sets form so-called \emph{$(m,\ell)$-covers} for $(m,\ell)=(2d-1,d)$ and $(m,\ell)=(2d-2,d)$, respectively, which gives an alternative way to finish the proofs, because the existence of an $(m,\ell)$-cover implies that the integrality gap is at most $\frac{m}{\ell}$~\cite{simultaneousMatchingArxiv}.
\end{remark}

\begin{remark}
  The local search algorithm given for the unweighted $d$-distance matching problem~\cite{madarasi2021matchings} also works for the unweighted $d$-distance $b$-matching problem.
  Therefore, we have a $(3/2+\varepsilon)$-approximation algorithm for the latter problem.
\end{remark}

\section{Hardness of approximation}\label{sec:distance2:hardness}
In the \emph{double matching problem}, we are given a bipartite graph $G=(S,T;E)$ and two sets $S_1,S_2\subseteq S$ such that ${S_1\cup S_2=S}$.
The goal is to find a maximum weight (size) subset $M$ of the edges for which both $M\cap E_1$ and $M\cap E_2$ are matchings, where $E_i$ denotes the edges induced by $S_i$ and $T$ for $i\in\{1,2\}$.
The double matching problem is known to be APX-hard in the weighted case~\cite{simultaneousMatchingArxiv}.
This implies that the weighted $d$-distance matching problem and the unweighted cyclic distance matching problems are APX-hard by a weight-preserving reduction from the double matching problem~\cite{madarasi2021matchings}.

In what follows, we prove that the hardness of approximation also applies to the unweighted non-cyclic case.

\begin{theorem}\label{thm:distance2:doubleUnweightedApxHard}
  The unweighted double matching problem is NP-hard to $\alpha$-approximate for any $\alpha < \frac{950}{949}$.
\end{theorem}
\begin{proof}
  Given are three finite disjoint sets $X,Y,Z$ and a set of hyperedges $\mc E\subseteq X\times Y\times Z$, a subset of the hyperedges $F\subseteq\mc E$ is called \emph{$3$-dimensional matching} if $x_1\neq x_2, y_1\neq y_2$ and $z_1\neq z_2$ for any two distinct triples $(x_1, y_1, z_1), (x_2, y_2, z_2) \in F$.
  Finding a maximum-size $3$-dimensional matching $F\subseteq\mc E$ cannot be approximated arbitrarily unless P=NP~\cite{kann1991maximumBoudned}.
  In fact, the problem remains NP-hard to approximate better than $\frac{95}{94}$ even for \emph{$2$-regular} instances, that is, when each element of $X\cup Y\cup Z$ occurs in exactly two triples in $\mc E$~\cite{CHLEBIK2006320}.
  To reduce the $2$-regular $3$-dimensional matching problem to the double matching problem, consider the following construction.

  Let $\mc{H}_X, \mc{H}_Y$ and $\mc{H}_Z$ denote three copies of the set of hyperedges $\mc H$, where the three versions of a hyperedge $e \in \mc{H}$ are $e^{(X)}\in\mc{H}_X,\ e^{(Y)}\in\mc{H}_Y$ and $e^{(Z)}\in\mc{H}_Z$.
  Define a bipartite graph $G=(S,T;E)$, where $S=\mc{H}_X \cup Y \cup \mc{H}_Z$, $T=X \cup \mc{H}_Y \cup Z$ and $E$ is as follows.
  For each $e \in \mc{H}$, add edges $e^{(X)}e^{(Y)}$ and $ e^{(Y)}e^{(Z)} \in E$, furthermore, add an edge to $G$ between $u$ and the two hyperedges in $\mc{H}_U$ incident to $u$ in $H$ for each $u \in U$, where $U \in \{X,Y,Z\}$.
  Let $S_1=\mc{H}_X \cup Y$ and $S_2=Y \cup \mc{H}_Z$.
  For a hyperedge $e=(x,y,z) \in \mc{H}$, let $K_e=\{e^{(X)}e^{(Y)}, e^{(Z)}e^{(Y)},xe^{(X)}, ye^{(Y)}, ze^{(Z)}\}\subseteq E$.
  Figures~\ref{fig:distance2:3dimMatching}~and~\ref{fig:distance2:3dimMatchingConstruction2} give an example for the construction.

  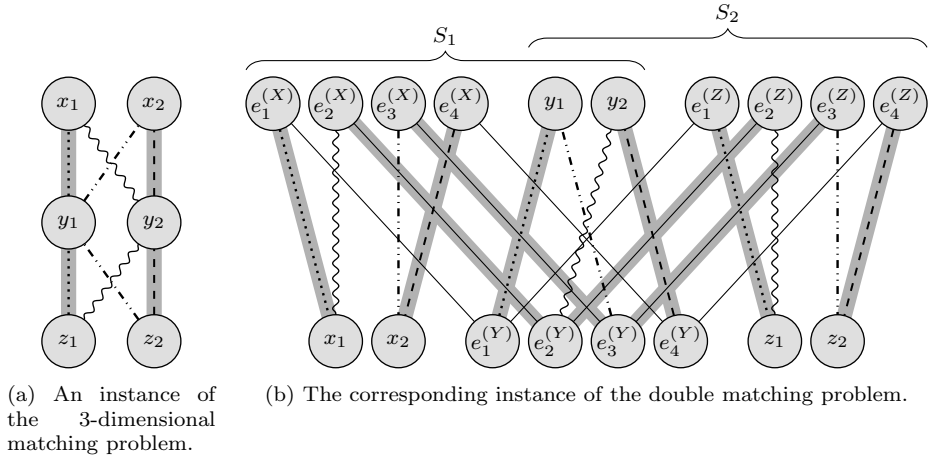
\begin{figure}
    \begin{subfigure}[t]{.225\textwidth}
      \centering
      \begin{tikzpicture}[scale=.75,xscale=1.5,yscale=1.2]
        \SetVertexMath
        \tikzset{VertexStyle/.append style = {minimum size = 20pt,inner sep=0pt}}
        \grEmptyPath[form=1,x=0,y=0,RA=1,rotation=0,prefix=z]{2}
        \grEmptyPath[form=1,x=0,y=1.75,RA=1,rotation=0,prefix=y]{2}
        \grEmptyPath[form=1,x=0,y=3.5,RA=1,rotation=0,prefix=x]{2}

        \draw[dotted,line width=.9] (z0) -- (y0);
        \draw[dotted,line width=.9] (y0) -- (x0);

        \draw[dashed,line width=.75] (z1) -- (y1);
        \draw[dashed,line width=.75] (y1) -- (x1);

        \draw[wavy] (z0) -- (y1);
        \draw[wavy] (y1) -- (x0);

        \draw[dashdotdotted,line width=.9] (z1) -- (y0);
        \draw[dashdotdotted,line width=.9] (y0) -- (x1);

        \begin{pgfonlayer}{background}
          \highlight{2mm}{black!30}{(z0.center) -- (y0.center) -- (x0.center)}
          \highlight{2mm}{black!30}{(z1.center) -- (y1.center) -- (x1.center)}
        \end{pgfonlayer}

      \end{tikzpicture}
      \caption{An instance of the $3$-dimensional matching problem.}\label{fig:distance2:3dimMatching}
    \end{subfigure}
    \hfill
    \begin{subfigure}[t]{.75\textwidth}
      \centering
      \begin{tikzpicture}[scale=.75,xscale=1.1,yscale=1.2]
        \SetVertexMath
        \tikzset{VertexStyle/.append style = {minimum size = 20pt,inner sep=0pt}}

        \grEmptyPath[form=1,x=0,y=3.5,RA=1,rotation=0,prefix=e^{(X)}]{4}
        \grEmptyPath[form=1,x=4.5,y=3.5,RA=1,rotation=0,prefix=y]{2}
        \grEmptyPath[form=1,x=7,y=3.5,RA=1,rotation=0,prefix=e^{(Z)}]{4}

        \grEmptyPath[form=1,x=1,y=0,RA=1,rotation=0,prefix=x]{2}
        \grEmptyPath[form=1,x=3.5,y=0,RA=1,rotation=0,prefix=e^{(Y)}]{4}
        \grEmptyPath[form=1,x=8,y=0,RA=1,rotation=0,prefix=z]{2}

        % \draw[dotted,line width=.9] (e^{(X)}0) -- (e^{(Y)}0);
        % \draw[dotted,line width=.9] (e^{(Z)}0) -- (e^{(Y)}0);
        \draw[] (e^{(X)}0) -- (e^{(Y)}0);
        \draw[] (e^{(Z)}0) -- (e^{(Y)}0);
        \draw[dotted,line width=.9] (e^{(X)}0) -- (x0);
        \draw[dotted,line width=.9] (e^{(Y)}0) -- (y0);
        \draw[dotted,line width=.9] (e^{(Z)}0) -- (z0);

        % \draw[wavy] (e^{(X)}1) -- (e^{(Y)}1);
        % \draw[wavy] (e^{(Z)}1) -- (e^{(Y)}1);
        \draw[] (e^{(X)}1) -- (e^{(Y)}1);
        \draw[] (e^{(Z)}1) -- (e^{(Y)}1);
        \draw[wavy] (e^{(X)}1) -- (x0);
        \draw[wavy] (e^{(Y)}1) -- (y1);
        \draw[wavy] (e^{(Z)}1) -- (z0);

        % \draw[dashed,line width=.75] (e^{(X)}2) -- (e^{(Y)}2);
        % \draw[dashed,line width=.75] (e^{(Z)}2) -- (e^{(Y)}2);
        \draw[] (e^{(X)}2) -- (e^{(Y)}2);
        \draw[] (e^{(Z)}2) -- (e^{(Y)}2);
        \draw[dashdotdotted,line width=.9] (e^{(X)}2) -- (x1);
        \draw[dashdotdotted,line width=.9] (e^{(Y)}2) -- (y0);
        \draw[dashdotdotted,line width=.9] (e^{(Z)}2) -- (z1);

        % \draw[dashdotdotted,line width=.9] (e^{(X)}3) -- (e^{(Y)}3);
        % \draw[dashdotdotted,line width=.9] (e^{(Z)}3) -- (e^{(Y)}3);
        \draw[] (e^{(X)}3) -- (e^{(Y)}3);
        \draw[] (e^{(Z)}3) -- (e^{(Y)}3);
        \draw[dashed,line width=.75] (e^{(X)}3) -- (x1);
        \draw[dashed,line width=.75] (e^{(Y)}3) -- (y1);
        \draw[dashed,line width=.75] (e^{(Z)}3) -- (z1);

        \begin{pgfonlayer}{background}
          \highlight{2mm}{black!30}{(e^{(X)}0.center) -- (x0.center)}
          \highlight{2mm}{black!30}{(e^{(Y)}0.center) -- (y0.center)}
          \highlight{2mm}{black!30}{(e^{(Z)}0.center) -- (z0.center)}

          \highlight{2mm}{black!30}{(e^{(X)}3.center) -- (x1.center)}
          \highlight{2mm}{black!30}{(e^{(Y)}3.center) -- (y1.center)}
          \highlight{2mm}{black!30}{(e^{(Z)}3.center) -- (z1.center)}

          \highlight{2mm}{black!30}{(e^{(X)}1.center) -- (e^{(Y)}1.center)}
          \highlight{2mm}{black!30}{(e^{(Z)}1.center) -- (e^{(Y)}1.center)}

          \highlight{2mm}{black!30}{(e^{(X)}2.center) -- (e^{(Y)}2.center)}
          \highlight{2mm}{black!30}{(e^{(Z)}2.center) -- (e^{(Y)}2.center)}
        \end{pgfonlayer}

        \draw [decorate,decoration={brace,amplitude=7pt,raise=1.5em}] (e^{(X)}0.west) -- (y1.east) node[midway,yshift=3em]{$S_1$};
        \draw [decorate,decoration={brace,amplitude=7pt,raise=2.5em}] (y0.west) -- (e^{(Z)}3.east) node[midway,yshift=4em]{$S_2$};
      \end{tikzpicture}
      \caption{The corresponding instance of the double matching problem.}\label{fig:distance2:3dimMatchingConstruction2}
    \end{subfigure}
    \caption{Illustration of the proof of Theorem~\ref{thm:distance2:doubleUnweightedApxHard}.
      Each hyperedge is represented by a unique line style.
      The highlighted $3$-dimensional matching in (a) corresponds to the highlighted solution in (b).}
  \end{figure}

  Assume that there exists an $\alpha$-approximation algorithm for the maximum double matching problem and let $M$ be an $\alpha$-approximate solution in $G$.
  We prove that one can construct a $(\frac{1}{10/\alpha-9})$-approximate $3$-dimensional matching in polynomial time using $M$, provided that $\alpha < \frac{10}{9}$.

  First, consider the following transformation of $M$.
  For each $e \in \mc{H}$, if $|K_e \cap M| < 3$, then add edges $e^{(X)}e^{(Y)}$ and $e^{(Y)}e^{(Z)}$ to $M$, and remove all other edges of $K_e$.
  After these operations, $M$ remains feasible and its size does not decrease, hence it remains an $\alpha$-approximate double matching.
  Observe that after the transformation we have either $|K_e \cap M|=2$ and hence $K_e \cap M = \{e^{(X)}e^{(Y)}, e^{(Y)}e^{(Z)}\}$, or $|K_e \cap M|=3$ and hence $K_e \cap M = \{xe^{(X)}, ye^{(Y)}, ze^{(Z)}\}$ for each $e=(x,y,z) \in \mc{H}$.

  Construct the $3$-dimensional matching $F\subseteq\mc H$ as the set of those hyperedges for which $|K_e \cap M|=3$.
  Note that $F$ is feasible, because $K_{e_1}\cap M= \{x_1e_1^{(X)}, y_1e_1^{(Y)}, z_1e_1^{(Z)}\}$ and $K_{e_2}\cap M= \{x_2e_2^{(X)}, y_2e_2^{(Y)}, z_2e_2^{(Z)}\}$ can hold simultaneously only if $x_1\neq x_2$, $y_1 \neq y_2$ and $z_1 \neq z_2$ --- as the degrees of these nodes are at most one in $M$.

  That is, we can construct a $3$-dimensional matching $F$ in $H$ such that
  \[
  |M|=3|F|+2(2|Z|-|F|) = |F|+4|Z|,
  \]
  since exactly three edges belong to each hyperedge in $F$, and two edges belong to each hyperedge in $\mc{H} \setminus F$.
  Applying this for a maximum double matching $M^*$, we get that
  \[|M^*| = |F'|+4|Z| \leq|F^*|+4|Z|,\]
  where $F'$ denotes the $3$-dimensional matching constructed from $M^*$, and $F^*$ is a maximum $3$-dimensional matching in $H$.
  Similarly, for any $3$-dimensional matching $F$ in $H$, we can create a double matching $M$ in $G$ such that $|M|=|F|+4|Z|$, so $|M^*| \geq |F^*|+4|Z|$, therefore
  \[
  |M^*| = |F^*|+4|Z|.
  \]
  Hence, for the $\alpha$-approximate double matching $M$, the $3$-dimensional matching $F$ constructed from $M$, and for optimal $M^*$ and $F^*$ solutions,
  \begin{equation}\label{eq:distance2:dmAPX:diff}
    |M^*|-|M| = |F^*|-|F|
  \end{equation}
  holds.

  With the greedy method, we can always construct a $3$-dimensional matching of size at least $\frac{|\mc{H}|}{4}$, therefore $|F^*|\geq\frac{|\mc{H}|}{4}=\frac{|Z|}{2}$.
  Using that $|S|=2|\mc{H}|+|Z|=5|Z|$ and that in a double matching $M$ each $s \in S$ has degree at most $1$, one gets that $|M^*| \leq 5|Z|$.
  Therefore,
  \begin{equation}\label{eq:distance2:dmAPX:boundOnOptM}
    |M^*| \leq 5|Z| \leq 10|F^*|
  \end{equation}
  holds.
  It follows from these observations that
  \begin{multline*}
    \frac{|F|}{|F^*|}
    = \frac{|F^*|-(|M^*|-|M|)}{|F^*|}
    = 1-\frac{|M^*|-|M|}{|F^*|}
    \geq 1-10\frac{|M^*|-|M|}{|M^*|}\\
    = 1-10\left(1-\frac{|M|}{|M^*|}\right)
    = 1-10\left(1-\frac{1}{\alpha}\right)
    = \frac{10}{\alpha}-9
  \end{multline*}
  where the first inequality follows from~(\ref{eq:distance2:dmAPX:diff}), and the second one holds by (\ref{eq:distance2:dmAPX:boundOnOptM}).
  If $\alpha < \frac{10}{9}$ and $F\neq\emptyset$, then $\frac{|F^*|}{|F|} \leq \frac{1}{10/\alpha-9}$, so if we had an $\alpha$-approximate algorithm for the double matching problem, then we could construct a $(\frac{1}{10/\alpha-9})$-approximate $3$-dimensional matching in polynomial time.
  The $2$-regular $3$-dimensional matching problem is NP-hard to $\beta$-approximate for any $\beta < \frac{95}{94}$, which implies that the double matching problem is NP-hard to $\alpha$-approximate for $\alpha<\frac{950}{949}$.
  \qed
\end{proof}

By the size-preserving reduction from the double matching problem to the distance matching problem given in~\cite{simultaneousMatchingArxiv}, the previous theorem implies the following.
\begin{theorem}\label{thm:distance2:distanceUnweightedApxHard}
  The unweighted distance matching problem is NP-hard to $\alpha$-approximate for any $\alpha < \frac{950}{949}$.
\end{theorem}

Clearly, this result also applies to the more general unweighted cyclic version of the problem.
Note that the proof given in~\cite{simultaneousMatchingArxiv} gives a slightly larger threshold of $\frac{760}{759}$ in the weighted non-cyclic and in the unweighted cyclic cases.

\section{Optimal permutations}\label{sec:distance2:optimalPermutations}

This section investigates a slightly different problem, which is motivated by the second application presented in the introduction.
It is a natural question whether we can find a permutation of $S$ --- which corresponds to the items on the conveyor belt --- maximizing the weight of the optimal $d$-distance $b$-matchings.
Formally, let $M_\pi^*$ denote an optimal $d$-distance $b$-matching under the permutation $\pi$ of $S$.
We want to find a permutation of $S$ and a $d$-distance $b$-matching $M^*$ with respect to this permutation such that $w(M^*)=\max_{\pi\in\symmgroup}w(M_\pi^*)$, where $\symmgroup$ is the set of all permutations of $S$.

In the next section, a polynomial-time algorithm is described for finding an optimal permutation and an optimal $d$-distance matching under this permutation (that is, when $b(s)=1$ for all $s\in S$ and $b(t)=\infty$ for all $t\in T$).
Section~\ref{sec:distance2:permDbMAndCDbM}, however, proves that the analogous problem is NP-hard for $d$-distance $b$-matchings even if $b\equiv 2$ or $d=2$, and gives $e$-approximation algorithms for general $b$ in both the cyclic and the non-cyclic cases, where $e$ is Euler's number.

As we have already seen, for a given permutation of $S$, it is NP-complete to decide whether a perfect (cyclic) $d$-distance matching exists; and finding a largest one is APX-hard.
In this light, it is quite surprising that we can find a permutation of $S$ which maximizes the weight of the maximum-weight (cyclic) $d$-distance matching --- furthermore, an optimal distance matching under the optimal permutation can be found as well.

Before entering the details of this method, we need the following lemma, which is easy to prove by a straightforward reduction to the circulation problem.
\begin{lemma}\label{lem:distance2:perm:degreeconstrainededgeset}
  For a bipartite graph $G=(S,T;E)$, a weight function $w: E\to \R_+$ on its edges and integers $k, r \in \Z_+$, we can find a maximum-weight subset of the edges in polynomial time satisfying the following three conditions:
  \begin{enumerate}[1)]
    \itemsep0em
  \item the degrees of the nodes in $S$ are at most $1$,
  \item the degrees of the nodes in $T$ are at most $(k+1)$, and
  \item there are at most $r$ nodes in $T$ with degree exactly $(k+1)$.
  \end{enumerate}
\end{lemma}

Next, we prove that one can find a permutation maximizing the size of the optimal $d$-distance matching.
Note that the first algorithm for the non-cyclic case appeared in~\cite{DmISCO2020}.
In the rest of this section, a revised, more intuitive approach is presented, which will be modified to handle the cyclic case as well.

\begin{theorem}\label{thm:distance2:perm:nonCyclicDMSolvable}
  For a bipartite graph $G=(S,T;E)$, a weight function $w: E\to \R_+$ on its edges and a positive integer $d\in\N$, we can find a permutation of $S$ along with a $d$-distance matching $M$ in polynomial time such that the weight of $M$ is the largest among all $d$-distance matchings over all permutations of $S$.
\end{theorem}
\begin{proof}
  Let $k,r\in\Z_+$ be such that $|S|=kd + r$, where $0 \leq r<d$.
  Find a maximum-weight edge set $M$ in $G$ such that the degrees of the nodes in $S$ are at most $1$, the degrees of the nodes in $T$ are at most $(k+1)$, and there are at most $r$ nodes in $T$ with degree exactly $(k+1)$.
  Such an edge set $M$ can be found in polynomial time by Lemma~\ref{lem:distance2:perm:degreeconstrainededgeset}.

  Clearly, a maximum-weight $d$-distance matchings under all permutations fulfill these three conditions, so for this largest possible weight $W$,
  \begin{equation}\label{eq:distance2:perm:noncyclicdmsolvable:mainineq}
    w(M) \geq W
  \end{equation}
  holds.
  To show equality, it suffices to construct a permutation of $S$ such that $M$ is a feasible $d$-distance matching in $G$.
  Let $t_1,\dots,t_{|T|}$ be a permutation of the nodes in $T$ which lists the nodes of degree $(k+1)$ first, then the nodes with degree smaller than $k$, and finally the nodes of degree $k$.
  Let $s'_1,\dots,s'_n$ be a permutation of $S$ in which the neighbors in $M$ of $t_j$ form an interval for all $j\in\{1,\dots,|T|\}$ and these intervals appear in the order given by $t_1,\dots,t_{|T|}$ (the order of the neighbors of any $t_j$ is arbitrary).
  Figure~\ref{fig:distance2:perm:nonCyclicDMSolvable:a} shows an example for the construction.
  Now, take a table of size $d\times (k+1)$, and remove all cells from the last column except for the first $r$.
  Fill the remaining cells of the table in row-major order with $s'_1,\dots,s'_n$.
  Let $s_1\dots,s_n$ be the permutation of $S$ obtained by reading the table in column-major order.
  Figure~\ref{fig:distance2:perm:nonCyclicDMSolvable:b} shows an example for the table-filling step.
  \begin{figure}[t]
    \begin{subfigure}[t]{.6\linewidth}
      \centering
      \begin{tikzpicture}[scale=.9,yscale=.9]
        \SetVertexMath
        \tikzset{VertexStyle/.append style = {minimum size = 16pt,inner sep=0pt}}

        \def\nodeDist{.7}

        \newcounter{cx}
        \forLoop{1}{11}{cx}{%
          \Vertex[x=\nodeDist*\thecx,y=2,L=s'_{\thecx}]{s_\thecx}
        }
        \forLoop{1}{6}{cx}{%
          \Vertex[x=2.5*\nodeDist+\nodeDist*\thecx,y=0,L=t_{\thecx}]{t_\thecx}
        }

        \draw (t_1) -- (s_1);
        \draw (t_1) -- (s_2);
        \draw (t_1) -- (s_3);

        \draw (t_2) -- (s_4);

        \draw (t_3) -- (s_5);

        \draw (t_4) -- (s_6);
        \draw (t_4) -- (s_7);

        \draw (t_5) -- (s_8);
        \draw (t_5) -- (s_9);

        \draw (t_6) -- (s_10);
        \draw (t_6) -- (s_11);

        \node(t1) at (3,-1.5) {$ $}; %invisible node to lift the graph
      \end{tikzpicture}
      \caption{A feasible matching for $d=4$.
        The indeces of the nodes correspond to the first step of the construction described in the proof.
        In this case, $k=2$ and $r=3$.}\label{fig:distance2:perm:nonCyclicDMSolvable:a}
    \end{subfigure}
    \hfill
    \begin{subfigure}[t]{0.36\linewidth}
      \centering
      \begin{tikzpicture}[scale=.8,yscale=.9,xscale=.9]
        \SetVertexMath
        \tikzset{VertexStyle/.append style = {minimum size = 16pt,inner sep=0pt}}

        \Vertex[x=1.5,y=3,L=s'_1]{a}
        \Vertex[x=3,y=3,L=s'_2]{e}
        \Vertex[x=4.5,y=3,L=s'_3]{i}

        \Vertex[x=1.5,y=1.5,L=s'_4]{b}
        \Vertex[x=3,y=1.5,L=s'_5]{f}
        \Vertex[x=4.5,y=1.5,L=s'_6]{j}

        \Vertex[x=1.5,y=0,L=s'_7]{c}
        \Vertex[x=3,y=0,L=s'_8]{g}
        \Vertex[x=4.5,y=0,L=s'_9]{k}

        \Vertex[x=1.5,y=-1.5,L=s'_{10}]{d}
        \Vertex[x=3,y=-1.5,L=s'_{11}]{h}

        \begin{pgfonlayer}{background}
          \highlight{9mm}{black!20}{(a.center) -- (i.center)}
          \highlight{9mm}{black!20}{(b.center) -- (b.center)}
          \highlight{9mm}{black!20}{(f.center) -- (f.center)}
          \highlight{9mm}{black!20}{(j.center) -- (j.center)}
          \fill [black!20] (4.5,1.5-.625) rectangle (4.5+.625,1.5+.625);
          \highlight{9mm}{black!20}{(c.center) -- (c.center)}
          \fill [black!20] (1.5-.625,-.625) rectangle (1.5,+.625);
          \highlight{9mm}{black!20}{(g.center) -- (k.center)}
          \highlight{9mm}{black!20}{(d.center) -- (h.center)}
        \end{pgfonlayer}

        % \draw[-{Stealth[length=2.5mm]}] (a) -- (b);
        % \draw[-{Stealth[length=2.5mm]}] (b) -- (c);
        % \draw[-{Stealth[length=2.5mm]}] (c) -- (d);
        % \draw[-{Stealth[length=2.5mm]}] (d) -- (e);
        % \draw[-{Stealth[length=2.5mm]}] (e) -- (f);
        % \draw[-{Stealth[length=2.5mm]}] (f) -- (g);
        % \draw[-{Stealth[length=2.5mm]}] (g) -- (h);
        % \draw[-{Stealth[length=2.5mm]}] (h) -- (i);
        % \draw[-{Stealth[length=2.5mm]}] (i) -- (j);
        % \draw[-{Stealth[length=2.5mm]}] (j) -- (k);

        \node(t1) at (3,3.6) {$t_1$};
        \node(t2) at (1.1,1.1) {$t_2$};
        \node(t3) at (2.6,1.1) {$t_3$};
        \node(t4) at (4.1,1.1) {$t_4$};
        \node(t4_) at (1.1,-.4) {$t_4$};
        \node(t5) at (3.75,-.4) {$t_5$};
        \node(t5) at (2.25,-1.85) {$t_6$};

      \end{tikzpicture}
      \caption{The table corresponding to the graph shown on the left.
        The optimal permutation is $s'_1,s'_4,s'_7,$  $s'_{10},s'_2,s'_5,s'_8,s'_{11},s'_3,s'_6,s'_9$}\label{fig:distance2:perm:nonCyclicDMSolvable:b}
    \end{subfigure}
    \caption{Illustration of the proof of Theorem~\ref{thm:distance2:perm:nonCyclicDMSolvable}.}
  \end{figure}
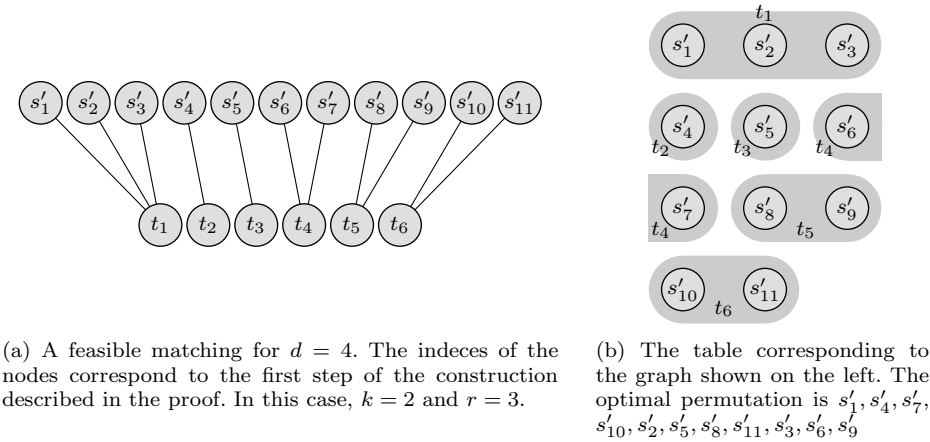

  We claim that $M$ is a feasible $d$-distance matching under the permutation $s_1,\dots,s_n$.
  The degrees in $M$ of the nodes in $S$ are clearly at most 1.
  To see that the distance constraints are met at each node $t\in T$, consider the following three cases.
  1) The degree of $t$ in $M$ is $(k+1)$.
  This means that the neighbors of $t$ occupy one of the first $r$ rows of the table, which are of length $(k+1)$.
  In column-major order, there are $(d-1)$ other nodes between any two consecutive neighbors of $t$, which was to be shown.
  2) The degree of $t$ in $M$ is smaller than $k$.
  The neighbors of $t$ are placed to at most two rows of the table in a row-major manner.
  Each of these rows is of length $k$ or $(k+1)$, therefore any two consecutive neighbors in row-major manner have at least $(d-1)$ other nodes between them.
  The first and the last neighbors of $t$ have a column between them which contains no neighbors of $t$, and hence there are more than $(d-1)$ nodes between them in column-major order.
  3) The degree of $t$ in $M$ is $k$.
  By construction, the neighbors of $t$ either occupy one of the last $(d-r)$ rows of the table --- which are of length $k$ ---, or they are placed to at most two rows in row-major manner such that the upper row is of length $(k+1)$.
  Similarly to the previous case, there are $(d-1)$ other nodes between any two consecutive neighbors and also between the first and the last one.
  These three cases prove that $M$ is feasible under the permutation $s_1,\dots,s_n$ of $S$.
  By~(\ref{eq:distance2:perm:noncyclicdmsolvable:mainineq}), this means that $M$ is a heaviest $d$-distance matching among all distance matchings under all permutations, which completes the proof of the theorem.
  \qed
\end{proof}

Now we prove the analogous theorem for the cyclic case.

\begin{theorem}\label{thm:distance2:perm:cyclicDMSolvable}
  For a bipartite graph $G=(S,T;E)$, a weight function $w: E\to \R_+$ on its edges and a positive integer $d\in\N$, we can find a permutation of $S$ and a maximum-weight cyclic $d$-distance matching $M$ with respect to this permutation in polynomial time such that the weight of $M$ is the largest over all permutations~of~$S$.
\end{theorem}
\begin{proof}
  We follow the same principle as in the proof of Theorem~\ref{thm:distance2:perm:nonCyclicDMSolvable}.
  Let $|S|=k d +r$, where $0\leq r<d$.
  Find a maximum-weight edge set $M$ in $G$ such that the degrees of the nodes in $S$ are at most $1$, and the degrees of the nodes in $T$ are at most $k$.
  Such an edge set can be found in polynomial time.
  Similarly to the proof of Theorem~\ref{thm:distance2:perm:nonCyclicDMSolvable}, it suffices to construct a permutation of $S$ such that $M$ is a feasible $d$-distance matching in $G$.

  Let $t_1,\dots,t_{|T|}$ be a permutation of the nodes in $T$ which lists the nodes of degree $k$ first, then the rest of the nodes.
  Let $s'_1,\dots,s'_n$ be a permutation of $S$ in which the neighbors in $M$ of $t_j$ form an interval for all $j\in\{1,\dots,k\}$ and these intervals appear in the order given by $t_1,\dots,t_{|T|}$ (the order of the neighbors of any $t_j$ is arbitrary).
  Now, take a table of size $(d+1)\times k$, and remove all cells from the last row except for the first $r$.
  Fill the remaining cells of the table in row-major order with $s'_1,\dots,s'_n$.
  Let $s_1\dots,s_n$ be the permutation of $S$ obtained by reading the table in column-major order.
  Similarly to the proof for the non-cyclic case, one can prove that $M$ is a feasible $d$-distance matching under this permutation.
  \qed
\end{proof}

Observe that Theorems~\ref{thm:distance2:perm:nonCyclicDMSolvable}~and~\ref{thm:distance2:perm:cyclicDMSolvable} extend to the case when degree bounds are also given for the nodes in $T$. To prove this, one can require that in the initial edge set $M$, found in the first steps of the proofs, the degree of each node $t\in T$ is also at most $b(t)$ --- Lemma~\ref{lem:distance2:perm:degreeconstrainededgeset} is easy to modify for finding such an edge set.

\subsection{(Cyclic) $d$-distance $b$-matchings}\label{sec:distance2:permDbMAndCDbM}
This section proves that the analogous problem for $d$-distance $b$-matchings is hard, even if $b\equiv 2$ or $d=2$.

\subsubsection{Hardness results}\label{sec:distance2:permDbMAndCDbMNPC}
We saw that an optimal permutation can be found in polynomial time for both cyclic and non-cyclic $d$-distance matchings.
This section investigates the complexity of the analogous problem for the $d$-distance $b$-matching problem.
First, we show that finding an optimal permutation is already hard when $b(s)=2$ for all $s\in S$ and $b(t)=\infty$ for all $t\in T$, that is, we consider the slight modification of the $d$-distance matching problem where the degree bound for each node in $S$ is two --- instead of the all-one bound.

\begin{theorem}\label{thm:distance2:perm:nonCyclicDbMNPC}
  It is NP-complete to decide whether there exists a permutation of $S$ such that there is a perfect $d$-distance $b'$-matching under this permutation, where $d=|S|/2$ and
  \begin{equation}\label{eq:distance2:perm:nonCyclicNPCDefOfB}
    b'(v)
    = \begin{cases}
      2&\text{ if } v\in S,\\
      \infty& \text{ if } v\in T
    \end{cases}
  \end{equation}
  for $v\in S\cup T$.
\end{theorem}
\begin{proof}
  In the \emph{$C_{4k+2}$Free2Factor problem}, a bipartite graph $G'=(S',T';E')$ is given and the goal is to decide whether it contains a $2$-factor (that is, a subgraph in which the degree of each node is exactly two) such that the length of every cycle in it is a multiple of $4$.
  This problem is known to be NP-complete~\cite{berczi2021complexity}, therefore, it suffices to reduce it to the problem defined in the theorem.

  Without loss of generality, we can assume that $|S'|=|T'|$ --- otherwise, the instance of the $C_{4k+2}$Free2Factor problem is not solvable.
  Let $G=(S,T;E)$ be a copy of $G'$, and add $|S'|$ new nodes to $S$, and $2|S|$ new nodes to $|T|$.
  Add $|S'|$ node-disjoint paths of length two on the newly added nodes such that the nodes in the middle of the paths are in $S$ and their endpoints are in $T$.
  We show that $G'$ has a $2$-factor consisting of cycles whose length is divisible by 4 if and only if there is a permutation of $S$ such that a perfect $d$-distance $b'$-matching exists in $G$, where $b'$ is as defined in~(\ref{eq:distance2:perm:nonCyclicNPCDefOfB}) and $d=|S|/2$.

  Let $s_1,\dots,s_{n}$ be a permutation of the nodes in $S$ such that there exists a perfect $d$-distance $b'$-matching $M \subseteq E$.
  The degrees in $M$ of the nodes in $S$ are exactly $2$, so $M$ contains all the edges of the paths of length $2$ added to $G$, and the degrees of the nodes in $T'\subseteq T$ are also exactly $2$ since $|S'|=|T'|$.
  Therefore, restricting $M$ to the edge set of the original graph $G'$, we get a $2$-factor in $G'$.
  We prove that the length of each cycle is a multiple of $4$.
  For all $t \in T$, if $s_it, s_jt \in M$ for some $i\neq j$ and hence $t \in T'$, then $|i-j| \geq |S'|$.
  But $|S|=2|S'|$, thus one of the indices $i$ and $j$ is in $\{1,\dots,|S|\}$, and the other one is in $\{|S'|+1,\dots,2|S'|\}$.
  This means that the nodes of $S$ can be divided into two disjoint sets $S_1$ and $S_2$ such that one of the neighbors of $t$ is in $S_1$ and the other one is in $S_2$ for all $t \in T'$.
  From this, one gets that every second node of any cycle in $M'$ is in $S$, and the nodes of the cycle in $S$ are alternately in $S_1$ and in $S_2$, so the length of every cycle must be a multiple of $4$.

  To finish the proof, we show that if there is a $2$-factor in $G'$ with cycles whose length is divisible by 4, then there exists a permutation of the nodes in $S$ such that there is a $d$-distance $b'$-matching in $G$.
  For each cycle of the $2$-factor, divide its nodes belonging to $S$ into two sets $S_1$ and $S_2$ alternately.
  Construct a permutation by enumerating the nodes of $S_1$ in arbitrary order, then the middle nodes of the paths of length two in arbitrary order, and finally the nodes of $S_2$ in arbitrary order.
  Under this permutation, the union of the edge set of the $2$-factor and the edges of the paths of length two form a $d$-distance $b'$-matching, where $b'$ is as defined in~(\ref{eq:distance2:perm:nonCyclicNPCDefOfB}) and $d=|S|/2=|S'|$.
  \qed
\end{proof}

To prove a similar theorem for the cyclic case, we need the following lemma.
\begin{lemma}\label{lem:distance2:vertexCoverByDisj4CyclesNPC}
  It is NP-complete to decide whether all nodes of a bipartite graph $G=(S,T;E)$ can be covered by node-disjoint cycles of length $4$.
\end{lemma}
\begin{proof}
  Given four disjoint sets $X_1,X_2,X_3,X_4$ and a set of hyperedges $\mc{E} \subseteq X_1 \times X_2 \times X_3 \times X_4$, a subset of the hyperedges $F~\subseteq~\mc{E}$ is called \emph{$4$-dimensional matching} if $u_1 \neq v_1, u_2 \neq v_2$, $u_3 \neq v_3$ and $u_4 \neq v_4$ for any two distinct hyperedges $(u_1,u_2,u_3,u_4)$, $(v_1,v_2,v_3,v_4) \in F$.
  It is NP-complete to decide whether there exists a $4$-dimensional matching of size $|X_1|$~\cite{Karp72}.

  We reduce the $4$-dimensional matching problem to the problem defined in the lemma statement.
  Let $H=(X_1 \cup X_2 \cup X_3 \cup X_4,\mc{E})$ denote the hypergraph given in the $4$-dimensional matching problem, and define an instance of the problem given in the statement of the lemma as follows.
  Let the node set of $G$ consist of the elements in $X_1 \cup X_2 \cup X_3 \cup X_4$ and, for each $e\in\mc{E}$, four additional nodes $v_1^e$, $v_2^e$, $v_3^e$ and $v_4^e$.
  For each hyperedge $(v_1,v_2,v_3,v_4) \in \mc{E}$ add the edges $v_1v_2, v_2v_3, v_3v_4$ and $v_4v_1$ to $G$.
  Also add the edges $v_1^ev_2^e$, $v_2^ev_3^e$, $v_3^ev_4^e$, $v_4^ev_1^e$, and $v_1v_1^e$, $v_2v_2^e$, $v_3v_3^e$, $v_4v_4^e$ to $G$ for all $e\in\mc{E}$.
  Figure~\ref{fig:distance2:vertexCoverByDisj4CyclesNPC:constrForHyperedge} illustrates the construction for the hyperedge $(v_1,v_2,v_3,v_4)\in\mc{E}$.
  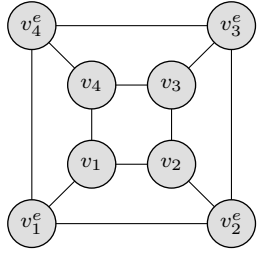
\begin{figure}
    \centering
    \begin{tikzpicture}[scale=.6]
      \SetVertexMath
      \tikzset{VertexStyle/.append style = {minimum size = 18pt,inner sep=0pt}}
      \Vertex[x=-1.75/2,y=-1.75/2,L=v_1]{v1}
      \Vertex[x=1.75/2,y=-1.75/2,L=v_2]{v2}
      \Vertex[x=1.75/2,y=1.75/2,L=v_3]{v3}
      \Vertex[x=-1.75/2,y=1.75/2,L=v_4]{v4}

      % inner cycle
      \draw[] (v1) -- (v2);
      \draw[] (v2) -- (v3);
      \draw[] (v3) -- (v4);
      \draw[] (v4) -- (v1);

      \begin{scope}[scale=2.5]
        \Vertex[x=-1.75/2,y=-1.75/2,L=v_1^e]{v1e}
        \Vertex[x=1.75/2,y=-1.75/2,L=v_2^e]{v2e}
        \Vertex[x=1.75/2,y=1.75/2,L=v_3^e]{v3e}
        \Vertex[x=-1.75/2,y=1.75/2,L=v_4^e]{v4e}
      \end{scope}

      % outer cycle
      \draw[] (v1e) -- (v2e);
      \draw[] (v2e) -- (v3e);
      \draw[] (v3e) -- (v4e);
      \draw[] (v4e) -- (v1e);

      % rays
      \draw[] (v1) -- (v1e);
      \draw[] (v2) -- (v2e);
      \draw[] (v3) -- (v3e);
      \draw[] (v4) -- (v4e);

    \end{tikzpicture}
    \caption{Illustration for the reduction in the proof of Lemma~\ref{lem:distance2:vertexCoverByDisj4CyclesNPC} for the hyperedge $(v_1,v_2,v_3,v_4)$.}\label{fig:distance2:vertexCoverByDisj4CyclesNPC:constrForHyperedge}
  \end{figure}
  It is easy to see that $G$ is bipartite: let $S$ consist of the nodes in $X_2 \cup X_4$ and also the nodes $v^e$ for $v \in X_1 \cup X_3$, $e\in\mc{E}$, and let $T$ consist of the rest of the nodes.
  By the definition of $E$, neither of the two parts induces any edges.
  We need to show that there is a perfect $4$-dimensional matching in $H$ if and only if $G$ has a vertex cover with node-disjoint $4$-cycles.

  Firstly, given a perfect $4$-dimensional matching $F$ in $H$, we construct the set of $4$-cycles as follows.
  For any hyperedge $e=(v_1,v_2,v_3,v_4) \in F$, select the cycle formed by the nodes $v_1,v_1^e,v_2^e,v_2$ and the cycle on nodes $v_3,v_3^e,v_4^e,v_4$.
  For each hyperedge $(u_1,u_2,u_3,u_4) \in \mc{E}\setminus F$, also add the cycle formed by the nodes $u_1^e,u_2^e,u_3^e,u_4^e$.
  The $4$-cycles obtained this way are pairwise node-disjoint, because the hyperedges in $F$ form a $4$-dimensional matching.
  Clearly, they cover all the nodes of $G$, because $F$ is a perfect $4$-dimensional matching.

  Secondly, assume that there exists a cover with node-disjoint $4$-cycles in $G$.
  We modify the set of the cycles as follows.
  If for a hyperedge $e=(v_1,v_2,v_3,v_4)$ both of the cycles formed by nodes $v_1,v_2,v_3,v_4$ and $v_1^e,v_2^e,v_3^e,v_4^e$ are in the subgraph, then delete these two cycles and add the cycles formed by the nodes $v_1,v_1^e,v_2^e,v_2$ and $v_3,v_3^e,v_4^e,v_4$ instead.
  Clearly, the new cycles are node-disjoint and cover every node of $G$.
  Now, construct a $4$-dimensional matching $F$ in $H$ by adding the hyperedge $e=(v_1,v_2,v_3,v_4)$ to $F$ if and only if the cycle formed by the nodes $v_1^e,v_2^e,v_3^e,v_4^e$ is not in the set of the selected $4$-cycles.
  This is a perfect $4$-dimensional matching, because if the cycle on nodes $v_1^e,v_2^e,v_3^e,v_4^e$ is not selected, then the two cycles covering these four nodes cover the nodes $v_1,v_2,v_3,v_4$, and because all nodes of $G$ are in exactly one of the selected $4$-cycles.
  \qed
\end{proof}

Now we are ready to prove the analogous theorem for the cyclic case.

\begin{theorem}\label{thm:distance2:perm:cyclicDbMNPC}
  It is NP-complete to decide whether there exists a permutation of the nodes in $S$ such that there is a perfect cyclic $d$-distance $b'$-matching under this permutation, where $d=|S|/2$ and
  \begin{equation}\label{eq:distance2:perm:cyclicNPCDefOfB}
    b'(v)
    = \begin{cases}
      2&\text{ if } v\in S,\\
      \infty& \text{ if } v\in T
    \end{cases}
  \end{equation}
  for $v\in S\cup T$.
\end{theorem}
\begin{proof}
  Clearly, the problem is in NP.
  We reduce the problem defined in Lemma~\ref{lem:distance2:vertexCoverByDisj4CyclesNPC} to the problem in the theorem.
  Let $G=(S,T;E)$ be a bipartite graph, and let $d=|S|/2$.
  Assume that $|S|=|T|$.
  We show that the nodes of $G$ can be covered by node-disjoint cycles of length $4$ if and only if there exists a permutation of the nodes in $S$ under which a perfect cyclic $d$-distance $b'$-matching exists, where $b'$ is as defined in~(\ref{eq:distance2:perm:cyclicNPCDefOfB}).
  This implies the statement of the theorem, since the former problem is NP-complete by Lemma~\ref{lem:distance2:vertexCoverByDisj4CyclesNPC}.

  Firstly, assume that we have a permutation $s_1,\dots,s_n$ of $S$ under which a perfect $d$-distance $b'$-matching $M$ exists.
  Since the $M$-degree of every node in $S$ is exactly $2$, the size of $M$ is $2|S|$.
  Furthermore, $d=|S|/2$, therefore the degree of $t$ in $M$ is at most $2$.
  From this, the degree of $t$ in $M$ is exactly $2$ by $|S|=|T|$, therefore $M$ is the node-disjoint union of cycles.
  If one of the neighbors of a node $s\in S$ is $t \in T$, then $t$ has a uniquely defined other neighbor, since there exists exactly one node in $S$ whose cyclic distance is at least $|S|/2$ from $s$.
  The same holds for the other neighbor of $s$.
  So, for every $s \in S$, the two neighbors of it have the same two neighbors, which gives a cycle of length $4$.
  Hence the edges of the cyclic $d$-distance $b'$-matching cover all nodes and is the union of node-disjoint $4$-cycles.

  Secondly, if there is a set of node-disjoint cycles of length $4$ covering the nodes of $G$, then we can construct a proper permutation of $S$ as follows.
  For each cycle, take one of its nodes in $S$, and put them in the first $|S|/2$ positions of the permutation in arbitrary order.
  This clearly determines the order of the rest of the nodes in $S$, since for any node $s$ from the first $|S|/2$ nodes, there is exactly one other node in $S$ having a common neighbor with $s$, so this must be placed exactly $|S|/2$ positions after $s$.
  \qed
\end{proof}

Note that the proofs of Theorems~\ref{thm:distance2:perm:nonCyclicDbMNPC}~and~\ref{thm:distance2:perm:cyclicDbMNPC} also show that the problem is NP-complete when $b(s)=b(t)=2$ for all $s\in S$ and $t\in T$ both in the cyclic and the non-cyclic case.
Next, Theorems~\ref{thm:distance2:perm:nonCyclicDbMNPC}~and~\ref{thm:distance2:perm:cyclicDbMNPC} are extended to the case when $b(s)\geq 2$ for all $s\in S$ --- instead of $b(s)=2$.

\begin{theorem}
  Both in the cyclic and the non-cyclic case, it is NP-complete to decide whether there exists a permutation of the nodes in $S$ such that there is a perfect (cyclic) $d$-distance $b''$-matching under this permutation, where $b''(s)\geq 2$ for all $s\in S$ and $b''(t)=\infty$ for all $t\in T$.
\end{theorem}
\begin{proof}
  In both cases, the problem is in NP.
  We prove the theorem for the cyclic and the non-cyclic versions simultaneously by showing that finding a permutation under which a perfect (cyclic) $d$-distance $b''$-matching exists is more general than finding a permutation under which a perfect (cyclic) $d$-distance $b'$-matching exists, where $b'$ is as defined in~(\ref{eq:distance2:perm:cyclicNPCDefOfB}).

  We modify the input graph of the perfect (cyclic) $d$-distance $b'$-matching problem as follows.
  For each $s \in S$, add $(b''(s)-2)$ new nodes to $T$, and connect them to $s$.
  Let $F$ denote the set of the newly added edges.
  We prove that there is a permutation of $S$ under which a perfect (cyclic) $d$-distance $b'$-matching exists in the original graph if and only if there is a perfect (cyclic) $d$-distance $b''$-matching in the new graph under the same permutation of $S$.

  Firstly, if there is a perfect (cyclic) $d$-distance $b'$-matching $M'$ in the original graph under some permutation of $S$, then $M'\cup F$ is a perfect (cyclic) $d$-distance $b''$-matching in the new graph under the same permutation.

  Secondly, for the reverse direction, let $M''$ denote a perfect (cyclic) $d$-distance $b''$-matching in the new graph under some permutation of $S$. Take the edge set $M''\setminus F$, and remove all except two of the edges incident to each node $s\in S$.
  This way one gets a feasible perfect (cyclic) $d$-distance $b'$-matching in the original graph under the same permutation.

  This completes the proof both in the cyclic and the non-cyclic version by Theorems~\ref{thm:distance2:perm:nonCyclicDbMNPC}~and~\ref{thm:distance2:perm:cyclicDbMNPC}, respectively.
  \qed
\end{proof}

\subsection{(Cyclic) $d$-distance $b$-matchings for small $d$}

By Theorems~\ref{thm:distance2:perm:nonCyclicDbMNPC}~and~\ref{thm:distance2:perm:cyclicDbMNPC}, finding an optimal permutation is hard when $d=|S|/2$.
We show that the problem is also hard for any $d\geq 2$ that is polynomially smaller than $|S|$, which, as a special case, implies that even the case $d=2$ is NP-complete.

\begin{theorem}\label{thm:distance2:perm:nonCyclicDbMNPC_d2}
  It is NP-complete to decide whether there exists a permutation of the nodes in $S$ such that there is a perfect (cyclic) $d$-distance $b'$-matching under this permutation, where $|S|=d(\ell+1)-2$, $d\geq 2$, $\ell=\Omega(|S|^c)$ for some constant $c>0$, and
  \begin{equation}\label{eq:distance2:perm:cyclicNPCDefOfB_d2}
    b'(v)
    = \begin{cases}
      \deg(v)&\text{ if } v\in S,\\
      \infty& \text{ if } v\in T
    \end{cases}
  \end{equation}
  for $v\in S\cup T$.
\end{theorem}
\begin{proof}
  The problem is clearly in NP.
  We give a reduction from the \emph{Hamiltonian path problem} in an undirected simple graph $G'=(V',E')$ to the non-cyclic case when $d=2$, then we show that the problem stated in the theorem includes this as a special case.
  Construct a bipartite graph $G=(S,T;E)$ such that $S=V'$ and $T$ is the edge set of the complement of $G'$.
  For each $t=uv\in T$, add edges $ut$ and $vt$ to $G$.
  We prove that there exists a Hamiltonian path in $G'$ if and only if there is a permutation of the nodes in $S$ under which a perfect $d$-distance $b'$-matching exists, where $d=2$ and $b'$ is as defined in~(\ref{eq:distance2:perm:cyclicNPCDefOfB_d2}).

  Firstly, assume that $M$ is a perfect $d$-distance $b'$-matching in $G$ under the permutation $s_1,\dots,s_n$ of $S$.
  By the definition of $b'$, this means that $M=E$, therefore $\{s_it,s_{i+1}t\}\not\subseteq E$ for all $i\in\{1,\dots,n-1\}$ and for all $t\in T$.
  But then $s_is_{i+1}\in E'$ for all $i\in\{1,\dots,n-1\}$, which means that $s_1,\dots,s_n$ defines a Hamiltonian path in $G'$.

  Secondly, assume that a Hamiltonian path in $G'$ traverses the nodes of $G'$ in the order $s_1,\dots,s_n$.
  This means that there is an edge $e_i$ in $E'$ between $s_i$ and $s_{i+1}$ for all $i\in\{1,\dots,n-1\}$, that is, $\{s_it,s_{i+1}t\}\not\subseteq E$ for all $i\in\{1,\dots,n-1\}$ and for all $t\in T$.
  This means that $M=E$ is a perfect $d$-distance $b'$-matching in $G$ under the permutation $s_1,\dots,s_n$ of $S$, where $d=2$ and $b'$ is as defined above.
  This completes the proof for $d=2$ in the non-cyclic case.\\

  Now we prove that if $n=|S|=d(\ell+1)-2$ for $\ell=\Omega(n^c)$ and $d\geq 2$, then the non-cyclic problem includes the case $d=2$ for $G'=(S',T';E')$ with $|S'|=2\ell$.
  Let $G=(S,T;E)$ be a copy of $G'$, and add a new node $\tilde s_{q,r}$ to $S$ for all $q\in\{1,\dots,\ell+1\}$ and $r\in\{1,\dots,d-2\}$.
  For $r\in\{1,\dots,d-2\}$, also add a new node $t_{r}$ to $T$ and all edges between $t_r$ and $\{\tilde s_{r,1},\dots,\tilde s_{r,\ell+1}\}$.

  Firstly, we show that if there is a permutation $s_1,\dots,s_n$ of $S$ under which $E$ is feasible in $G$, then there is a permutation of $S'$ under which $E'$ is feasible in $G'$.
  We claim that $s_{id-1},s_{id}\in S'$ for all $i\in\{1,\dots,\ell\}$.
  By contradiction, suppose that $s_{id-1}=\tilde s_{q,r}$ or $s_{id}=\tilde s_{q,r}$ for some $q$ and $r$.
  Let $t_r$ denote the only neighbor of $\tilde s_{q,r}$.
  As $E$ is feasible, all neighbors of $t_r$ on the left of $\tilde s_{q,r}$ are among the nodes $s_1,\dots,s_{id-d}$.
  Similarly, all neighbors of $t_r$ on the right of $\tilde s_{q,r}$ are among the nodes $s_{id-1+d},\dots,s_{n}$.
  Therefore, the degree of $t$ must be at most
  \[\left\lceil \frac{id-d}{d} \right\rceil+1+\left\lceil \frac{n-(id-1+d)+1}{d} \right\rceil = \ell,\]
  which contradicts the fact that the degree of $t_r$ is $(\ell+1)$ by the construction of $G$.
  This means that $s_{id-1},s_{id}\in S'$ for all $i\in\{1,\dots,\ell\}$, as we claimed.
  Since the size of $S'$ is $2\ell$, the nodes $s_{id-1},s_{id}$ for $i\in\{1,\dots,\ell\}$ are exactly the nodes in $S'$.
  Therefore, every interval of length $d$ contains two nodes in $S'$, which means that $E'$ is a feasible solution for $G'$ and $d=2$ under the permutation of $S'$ obtained by restricting $s_1,\dots,s_n$ to $S'$.

  Secondly, we prove that if there is a permutation $s'_1,\dots,s'_{|S'|}$ of $S'$ under which $E'$ is feasible, then there is a permutation of $S$ under which $E$ is feasible.
  Concatenate the subsequences $\tilde s_{q,1},\dots,\tilde s_{q,d-2}$ for $q=1,\dots,\ell+1$.
  Then, insert $s'_{2q-1}$ and $s'_{2q}$ in this order right after $\tilde s_{q,d-2}$ for $q\in\{1,\dots,\ell\}$.
  We claim that $E$ is feasible under the permutation obtained this way.
  As an interval of length $d$ includes exactly two nodes in $S'$ and these nodes appear in the order given by $s'_1,\dots,s'_{|S'|}$, the edge set $E'$ is feasible.
  For each $r\in\{1,\dots,d-2\}$, the neighbors of $t_r$ are the nodes $\tilde s_{1,r},\dots,\tilde s_{\ell+1,r}$, which do not appear in an interval of length $d$ in the permutation defined above.
  This means that the edge set $E\setminus E'$ is feasible as well.
  Since $E'$ and $E\setminus E'$ are node-disjoint, this implies that $E$ is also feasible, which was to be shown.
  This completes the proof in the non-cyclic case.
  \\

  To show the hardness of the cyclic case, one can give a reduction from the \emph{Hamiltonian cycle problem}.
  The proof is a straightforward modification of the reduction for the non-cyclic case, therefore it is left to the reader.
  \qed
\end{proof}

In the previous theorem, we did not assume that the coordinates of $b$ are small, unlike in Theorems~\ref{thm:distance2:perm:nonCyclicDbMNPC}~and~\ref{thm:distance2:perm:cyclicDbMNPC}.
It remains open whether the problem becomes tractable when $d=2$ and the coordinates of $b$ are small, for example $b\equiv 2$.

\subsection{Approximation algorithms for finding a permutation}
In this section, we give $e$-approximation algorithms for finding the best (cyclic) permutation under which the weight of the optimal (cyclic) $d$-distance $b$-matching is as large as possible.
The approximation algorithms also give an $e$-approximate (cyclic) $d$-distance $b$-matching under the permutation.
Both algorithms are randomized, but they are easy to de-randomize, which we briefly discuss at the end of the section.
\begin{algorithm}[t]
  % \caption{\hspace{0.5cm}\textsc{Randomized approximate cyclic $S$-permutation}}\label{alg:distance2:sillyCyclicRandPerm}
  \caption{\hspace{0.5cm}{Randomized approximation algorithm for cyclic $S$-permutations}}\label{alg:distance2:sillyCyclicRandPerm}
  \begin{algorithmic}
    % \State Let $G=(S,T;E)$, $w:E\longrightarrow\R$ and $d\in\N$.
    \State For $v\in S\cup T$, let $b'(v)=\begin{cases}
      b(v) & \text{if } v\in S,\\
      \min\{\floor{\frac{n}{d}},b(t)\} & \text{if } v\in T.
    \end{cases}$
    \State Find a maximum-weight $b'$-matching $\widehat M$ in $G$.
    \State Generate a cyclic permutation $s_1,\dots,s_n$ of $S$ uniformly at random.
    \State $M:=\emptyset$
    \For{$t\in T$}
    \For{$i=1,\dots,n$}
    \If{$s_it\in\widehat M$\text{ and }$s_{i-d+1}t,\dots,s_{i-1}t\notin\widehat M$}
    \State $M:=M\cup\{s_it\}$
    \EndIf
    \EndFor
    \EndFor
    \State \textbf{output} $s_1,\dots,s_n$ and $M$
  \end{algorithmic}
\end{algorithm}

First, consider the cyclic version of the problem.
Algorithm~\ref{alg:distance2:sillyCyclicRandPerm} finds a maximum-weight $b'$-matching $\widehat{M}$ in $G$ for the $b'$ defined in the algorithm
and takes a random cyclic permutation of $S$.
Then, for each $t\in T$, it adds an edge $s_it\in\widehat M$ to the solution if and only if $t$ and the $(d-1)$ nodes cyclically before $s_i$ induce no edges in $\widehat M$.
The algorithm returns the chosen permutation and the union of the selected edges.
Clearly, this edge set is a feasible cyclic $d$-distance $b$-matching under the chosen permutation, because 1) the degree of any node $s\in S$ is at most $b(s)$, since the found edge set is a subset of $\widehat M$, and 2) the distance constraints are met at each node $t\in T$, since no edge $st$ is added for which an edge $s't\in\widehat M$ exists such that $s'$ is one of the $(d-1)$ nodes before $s$ cyclically.

The following theorem gives a lower bound on the expected weight of the solution found by Algorithm~\ref{alg:distance2:sillyCyclicRandPerm}.
\begin{theorem}\label{thm:distance2:randCyclicPerm:apxMain}
  Algorithm~\ref{alg:distance2:sillyCyclicRandPerm} outputs a cyclic permutation $s_1,\dots,s_n$ of $S$ and a feasible cyclic $d$-distance $b$-matching $M$ whose expected weight is at least
  \begin{equation}\label{eq:distance2:randCyclicPerm:apxMain}
    \max\left\{\left(1-\frac{1}{d}\right)^{d-1}, \left(1-\frac{1}{k}\right)^{k-1}\right\}
  \end{equation}
  times the weight of the heaviest cyclic $d$-distance $b$-matching under all permutations, where $k=\max_{t\in T}{b'(t)}$ for the $b'$ defined in Algorithm~\ref{alg:distance2:sillyCyclicRandPerm}.
  This lower bound is tight.
\end{theorem}
\begin{proof}
  As we have already seen, Algorithm~\ref{alg:distance2:sillyCyclicRandPerm} returns a feasible solution under the randomly chosen permutation.
  Let $\E(n,d,k)$ denote the expected weight of the solution found by the algorithm, and let $\alpha(d,k)$ denote the lower bound given by~(\ref{eq:distance2:randCyclicPerm:apxMain}).
  First, we consider the unweighted case when $T=\{t\}$.
  Without loss of generality, one can assume that $k\geq 2$ and $d\leq n$.
  Note that $n\geq dk$ holds by the definition of $b'$.
  Let $P(n,d,k)$ be the probability that a given edge is added to the solution.
  By definition,
  \begin{equation}\label{eq:distance2:randCyclicPerm:defP}
    P(n,d,k)
    =\prod_{i=1}^{d-1}\frac{n-k-(i-1)}{n-i}.
  \end{equation}
  Clearly, $\E(n,d,k)=kP(n,d,k)$.
  Observe that
  \begin{equation}\label{eq:distance2:randCyclicPerm:PLargestWhenNIsSpec}
    P(n,d,k)\geq P(dk,d,k)
  \end{equation}
  holds for all $n,d,k\in\N$ provided that $n\geq dk$, because $dk$ is the smallest possible size of $S$ when the degree of $t$ can be $k$ in $\widehat M$, and if the number of nodes in $S$ is larger than $dk$, then all additional nodes must be isolated, hence the probability of an edge being added can be only larger.
  The value of $P(dk,d,k)$ can be expressed as follows.
  \begin{multline*}
    P(dk,d,k)
    =\prod_{i=1}^{d-1}\frac{kd-k-i+1}{kd-i}
    =\frac{\prod\limits_{j=\max\{1,d-k+1\}}^{d-1}kd-k-j+1}{\prod\limits_{i=1}^{\min\{k-1,d-1\}}kd-i}\\
    =\prod\limits_{i=1}^{\min\{k-1,d-1\}}\frac{kd-k-(d-i)+1}{kd-i}
    =\prod\limits_{i=1}^{\min\{k-1,d-1\}}\frac{(k-1)d-k+i+1}{kd-i},
  \end{multline*}
  where the first equation holds by~(\ref{eq:distance2:randCyclicPerm:defP}), and the third one because the product is telescopic.
  From this, we immediately get that $P(dk,d,k)$ is monotone decreasing in $d$ for all $k\in\N$, and it tends to $\left(1-\frac{1}{k}\right)^{k-1}$ as $d$ goes to infinity for all $k\in\N$.
  This implies that
  \begin{equation}\label{eq:distance2:randCyclicPerm:LimitInD}
    P(dk,d,k)\geq \left(1-\frac{1}{k}\right)^{k-1}
  \end{equation}
  holds for all $d,k\in\N$.
  Similarly, $P(dk,d,k)$ is monotone decreasing in $k$ for all $d\in\N$, and it tends to $\left(1-\frac{1}{d}\right)^{d-1}$ as $k$ goes to infinity for all $k\in\N$.
  This implies that
  \begin{equation}\label{eq:distance2:randCyclicPerm:LimitInK}
    P(dk,d,k)\geq \left(1-\frac{1}{d}\right)^{d-1}
  \end{equation}
  holds for all $d,k\in\N$.

  By~(\ref{eq:distance2:randCyclicPerm:PLargestWhenNIsSpec}),~(\ref{eq:distance2:randCyclicPerm:LimitInD})~and~(\ref{eq:distance2:randCyclicPerm:LimitInK}), $P(n,d,k)\geq\alpha(d,k)$
  holds for all $n,d,k\in\N$ provided that $n\geq dk$, which completes the proof of the unweighted case when $|T|=1$.
  The bounds given by~(\ref{eq:distance2:randCyclicPerm:LimitInD})~and~(\ref{eq:distance2:randCyclicPerm:LimitInK}) are (asymptotically) tight, therefore~(\ref{eq:distance2:randCyclicPerm:apxMain}) cannot be improved, as we stated in the theorem.

  Since all edges incident to $t$ appear in the output of Algorithm~\ref{alg:distance2:sillyCyclicRandPerm} with equal probability, the expected weight of the returned edges is at least $\alpha(d,k)w(\widehat M\cap\Delta(t))$, which was to be shown in the weighted case when $|T|=1$.\\

  We continue with the general weighted case, that is, when the size of $T$ is arbitrary.
  Let $\cyclicSymmgroup$ denote the set of all cyclic permutations of $S$, and let $M_{\pi}\subseteq\widehat M$ denote the feasible cyclic $d$-distance $b$-matching returned by the algorithm when it selects the cyclic permutation $\pi\in\cyclicSymmgroup$.
  The following computation leads to the bound stated in the theorem.
  \begin{multline*}
    \E(n,d,k)
    =\frac{\sum_{\pi\in\cyclicSymmgroup}w(M_{\pi})}{|\cyclicSymmgroup|}
    = \frac{\sum_{\pi\in\cyclicSymmgroup}\sum_{t\in T}w(M_{\pi}\cap \Delta(t))}{|\cyclicSymmgroup|}\\
    = \sum_{t\in T}\frac{\sum_{\pi\in\cyclicSymmgroup}w(M_{\pi}\cap \Delta(t))}{|\cyclicSymmgroup|}
    = \sum_{t\in T}P(n,d,\deg_{\widehat M}(t))w(\widehat M\cap\Delta(t))\\
    \geq \sum_{t\in T}\alpha(d,k)w(\widehat M\cap\Delta(t))
    =\alpha(d,k)w(\widehat M)
    \geq\alpha(d,k)w(M^*),
  \end{multline*}
  where $M^*$ is an optimal cyclic $d$-distance $b$-matching under all permutations.
  The first inequality holds by the case $|T|=1$, and the second one because any cyclic $d$-distance $b$-matching must respect the degree bounds posed by $b'$ and $\widehat M$ is a heaviest $b'$-matching for the $b'$ defined in Algorithm~\ref{alg:distance2:sillyCyclicRandPerm}.
  This completes the proof of the theorem.
  \qed
\end{proof}

It is well known that~(\ref{eq:distance2:randCyclicPerm:apxMain}) is monotone decreasing and tends to $\frac{1}{e}$ as $k$ and $d$ go to infinity.
This immediately implies the following.

\begin{theorem}
  Algorithm~\ref{alg:distance2:sillyCyclicRandPerm} outputs a cyclic permutation $s_1,\dots,s_n$ of $S$ and a feasible cyclic $d$-distance $b$-matching whose expected weight is at least $\frac{1}{e}$ times the weight of the heaviest cyclic $d$-distance matching under all permutations.
  This lower bound is tight.
\end{theorem}

Now, we turn to the non-cyclic case.
\begin{algorithm}
  \caption{\hspace{0.5cm}{Randomized approximation algorithm for $S$-permutations}}\label{alg:distance2:sillyRandPerm}
  \begin{algorithmic}
    \State For $v\in S\cup T$, let $b'(v)=\begin{cases}
      b(v) & \text{if } v\in S,\\
      \min\{\ceil{\frac{n}{d}},b(t)\} & \text{if } v\in T.
    \end{cases}$
    \State Find a maximum-weight $b'$-matching $\widehat M$ in $G$.
    \State Generate a permutation $s_1,\dots,s_n$ of $S$ uniformly at random.
    \State $M:=\emptyset$
    \For{$t\in T$}
    \For{$i=1,\dots,n$}
    \If{$s_it\in\widehat M$\text{ and }$s_{\max\{1,{i-d+1}\}}t,\dots,s_{i-1}t\notin\widehat M$}
    \State $M:=M\cup\{s_it\}$
    \EndIf
    \EndFor
    \EndFor
    \State \textbf{output} $s_1,\dots,s_n$ and $M$
  \end{algorithmic}
\end{algorithm}

Algorithm~\ref{alg:distance2:sillyRandPerm} is the analog of Algorithm~\ref{alg:distance2:sillyCyclicRandPerm} for the non-cyclic $d$-distance $b$-matching problem.
First, the algorithm finds a maximum-weight $b'$-matching $\widehat{M}$ in $G$ for the $b'$ defined in Algorithm~\ref{alg:distance2:sillyRandPerm},
and takes a random permutation of $S$.
Then, for each $t\in T$, it selects an edge $s_it\in\widehat M$ if $t$ and the (at most) $(d-1)$ nodes before $s_i$ induce no edges in $\widehat M$.
It returns the chosen permutation and the union of the selected edges.
Similarly to Algorithm~\ref{alg:distance2:sillyCyclicRandPerm}, the edge set returned by the algorithm is a feasible $d$-distance $b$-matching under the chosen permutation.\\

The following theorem for the non-cyclic version is analogous to Theorem~\ref{thm:distance2:randCyclicPerm:apxMain}.

\begin{theorem}\label{thm:distance2:randPerm:apxMain}
  Algorithm~\ref{alg:distance2:sillyRandPerm} outputs a permutation $s_1,\dots,s_n$ of $S$ and a feasible $d$-distance $b$-matching whose expected weight is at least
  \begin{equation}\label{eq:distance2:randPerm:apxMain}
    \max\left\{\left(1-\frac{1}{d}\right)^{d-1}, \frac{1+(k-1)\left(1-\frac{1}{k-1}\right)^k}{k}\right\}
  \end{equation}
  times the weight of the heaviest $d$-distance $b$-matching under all permutations, where $k=\max_{t\in T}{b'(t)}$ for the $b'$ defined in Algorithm~\ref{alg:distance2:sillyRandPerm}.
  This lower bound is tight.
\end{theorem}
\begin{proof}
  The outline of the proof is similar to that of Theorem~\ref{thm:distance2:randCyclicPerm:apxMain}, but the technical details are slightly more complicated.
  Let $\E(n,d,k)$ denote the expected weight of the solution returned by the algorithm, and let $\beta(d,k)$ denote the lower bound given by~(\ref{eq:distance2:randPerm:apxMain}).
  Similarly to the proof of Theorem~\ref{thm:distance2:randCyclicPerm:apxMain}, consider the case when $T=\{t\}$ and the problem is unweighted.
  Without loss of generality, one can assume that $k\geq 2$ and $d\leq n$.
  Let $P(n,d,k)$ be the probability that a given edge is added to the solution, and let $P(n,d,k,i)$ denote the probability that the edge incident to $s_i$ is added to the solution.
  By definition,
  \begin{equation*}
    P(n,d,k)
    = \frac{1}{n}\sum_{i=1}^{n}P(n,d,k,i),
  \end{equation*}
  and
  \begin{equation*}
    P(n,d,k,i)
    =\prod_{j=1}^{\min\{d-1,i-1\}}\frac{n-k-(j-1)}{n-j}.
  \end{equation*}
  Clearly, $\E(n,d,k)=kP(n,d,k)$.
  Observe that
  \begin{equation*}
    P(n,d,k)\geq P((d-1)k+1,d,k)
  \end{equation*}
  holds for all $n,d,k\in\N$ provided that  $n\geq (d-1)k+1$, because $((d-1)k+1)$ is the smallest possible size of $S$ when the degree of $t$ can be $k$ in $\widehat M$, and if the number of nodes in $S$ is larger, then all the extra nodes are isolated, hence the probability that an edge is added can be only larger.
  Let $\bar{n}=(d-1)k+1$.
  The value of $P(\bar{n},d,k)$ can be expressed as follows.
  \begin{multline}\label{eq:distance2:randPerm:PExpr}
    P(\bar{n},d,k)
    = \frac{1}{\bar{n}}\sum_{i=1}^{\bar{n}}P(\bar{n},d,k,i)
    = \frac{1}{\bar{n}}\sum_{i=1}^{\bar{n}}\prod_{j=1}^{\min\{d-1,i-1\}}\frac{\bar{n}-k-(j-1)}{\bar{n}-j}\\
    = \frac{1}{\bar{n}}\sum_{i=1}^{d-1}\prod_{j=1}^{i-1}\frac{\bar{n}-k-(j-1)}{\bar{n}-j} + \frac{\bar{n}-d+1}{\bar{n}}\prod_{j=1}^{d-1}\frac{\bar{n}-k-(j-1)}{\bar{n}-j}\\
    = \frac{1}{(k-1)d+1}\sum_{i=1}^{d-1}\prod_{j=1}^{i-1}\frac{(k-1)d-k-j+2}{(k-1)d+1-j}\\
    + \frac{(k-1)d-d+2}{(k-1)d+1}\prod_{j=1}^{d-1}\frac{(k-1)d-k-j+2}{(k-1)d+1-j}\\
    = \frac{1}{(k-1)d+1}\sum_{i=1}^{d-1}\prod_{j=1}^{\min\{i-1,k-1\}}\frac{(k-1)d+j-i-k+2}{(k-1)d-j+1}\\
    + \frac{(k-2)d+2}{(k-1)d+1}\prod_{j=1}^{\min\{k-1,d-1\}}\frac{(k-2)d-k+j+2}{(k-1)d-j+1},
  \end{multline}
  where the last equation holds by rearranging the products.
  Let $f(d,k)$ and $g(d,k)$ denote the first and the second summand in the right hand-side of (\ref{eq:distance2:randPerm:PExpr}), respectively.
  Clearly,
  \begin{equation}\label{eq:distance2:randPerm:limitOfG}
    \lim_{d\to\infty}g(d,k)
    =\lim_{d\to\infty}\frac{(k-2)d+2}{(k-1)d+1}\prod_{j=1}^{\min\{k-1,d-1\}}\frac{(k-2)d-k+j+2}{(k-1)d-j+1}\\
    =\left(\frac{k-2}{k-1}\right)^k.
  \end{equation}
  To derive the limit of $f(d,k)$, we need the following computation.
  \begin{multline}\label{eq:distance2:randPerm:computeF}
    ((k-1)d+1)f(d,k)
    =\sum_{i=1}^{d-1}\prod_{j=1}^{\min\{i-1,k-1\}}\frac{(k-1)d+j-i-k+2}{(k-1)d-j+1}\\
    =\sum_{i=1}^{d-1}\prod_{j=1}^{\min\{i-1,k-1\}}\frac{(k-1)d-k-j+2}{(k-1)d-j+1}
    =\frac{\sum_{i=1}^{d-1}{(k-1)d-i+1 \choose k-1}}{{(k-1)d \choose k-1}}\\
    =\frac{{(k-1)d+1 \choose k} - {(k-2)d+2 \choose k}}{{(k-1)d \choose k-1}}
    =\frac{((k-1)d+1){(k-1)d \choose k-1} - ((k-2)d+2){(k-2)d+1 \choose k-1}}{k{(k-1)d \choose k-1}}\\
    =\frac{(k-1)d+1}{k}-\frac{((k-2)d+2){(k-2)d+1 \choose k-1}}{k{(k-1)d \choose k-1}}\\
    =\frac{(k-1)d+1}{k}-\frac{(k-2)d+2}{k}\prod_{j=1}^{k-1}\frac{(k-2)d-j+2}{(k-1)d-j+1},
  \end{multline}
  where the first equation holds by the definition of $f$, the second one by rearranging the product, and the fourth one by applying the binomial identity ${\sum_{q=0}^{N}{q \choose K}={N+1 \choose K+1}}$ twice.
  Using~(\ref{eq:distance2:randPerm:computeF}), we get that
  \begin{multline}\label{eq:distance2:randPerm:limitOfF}
    \lim_{d\to\infty}f(d,k)
    =\lim_{d\to\infty}\frac{1}{k}-\frac{(k-2)d+2}{k((k-1)d+1)}\prod_{j=1}^{k-1}\frac{(k-2)d-j+2}{(k-1)d-j+1}\\
    =\frac{1-\left(\frac{k-2}{k-1}\right)^k}{k}
  \end{multline}
  for all $k\in\N$.
  By~(\ref{eq:distance2:randPerm:PExpr}),~(\ref{eq:distance2:randPerm:limitOfG})~and~(\ref{eq:distance2:randPerm:limitOfF}), $P(\bar{n},d,k)$ tends to
  \begin{equation*}
    \left(\frac{k-2}{k-1}\right)^k+\frac{1-\left(\frac{k-2}{k-1}\right)^k}{k}=\frac{1+(k-1)\left(\frac{k-2}{k-1}\right)^k}{k}
  \end{equation*}
  as $d$ goes to infinity for all $k\geq 2$.
  Observe that $P(n,d,k,i)$ is non-increasing in $d$, and therefore so is $P(n,d,k)$ for all $n,k\in\N$.
  This implies that
  \begin{equation}\label{eq:distance2:randPerm:PBound1}
    P(n,d,k)\geq\frac{1+(k-1)\left(\frac{k-2}{k-1}\right)^k}{k}
  \end{equation}
  holds for all $n,d,k\in\N$ provided that $n\geq (k-1)d+1$.
  Similarly, $P(n,d,k)$ is non-increasing in $k$ for all $n,d\in\N$.
  From~(\ref{eq:distance2:randPerm:PExpr}), it is easy to see that $P(\bar{n},d,k)$ tends to $\left(1-\frac{1}{d}\right)^{d-1}$ as $k$ goes to infinity, therefore,
  \begin{equation}\label{eq:distance2:randPerm:PBound2}
    P(n,d,k)\geq\left(1-\frac{1}{d}\right)^{d-1}
  \end{equation}
  holds.
  By~(\ref{eq:distance2:randPerm:PBound1})~and~(\ref{eq:distance2:randPerm:PBound2}), $P(n,d,k)\geq\alpha(d,k)$
  follows for all $n,d,k\in\N$ provided that $n\geq (k-1)d+1$, which completes the proof of the unweighted case when $|T|=1$.
  The bounds given in~(\ref{eq:distance2:randPerm:PBound1})~and~(\ref{eq:distance2:randPerm:PBound2}) are (asymptotically) tight, therefore~(\ref{eq:distance2:randPerm:apxMain}) cannot be improved, as we stated in the theorem.

  Since all edges incident to $t$ appear in the output of Algorithm~\ref{alg:distance2:sillyRandPerm} with equal probability, the expected weight of the returned edges is at least $\beta(d,k)w(\widehat M\cap\Delta(t))$, which was to be shown in the weighted case when $|T|=1$.\\

  The general weighted case, when the size of $T$ is arbitrary, can be handled by a computation similar to the end of the proof of~Theorem~\ref{thm:distance2:randCyclicPerm:apxMain}.
  From this, we get that $\E(n,d,k)\geq\beta(d,k)w(M^*)$, which completes the proof.
  \qed
\end{proof}

It is easy to see that~(\ref{eq:distance2:randPerm:apxMain}) is monotone decreasing and tends to $\frac{1}{e}$ as $k$ and $d$ go to infinity.
This immediately implies the following.

\begin{theorem}
  Algorithm~\ref{alg:distance2:sillyRandPerm} outputs a permutation $s_1,\dots,s_n$ of $S$ and a feasible $d$-distance $b$-matching whose expected weight is at least $\frac{1}{e}$ times the weight of the heaviest $d$-distance matching under all permutations.
  This lower bound is tight.
\end{theorem}

By Theorems~\ref{thm:distance2:randCyclicPerm:apxMain}~and~\ref{thm:distance2:randPerm:apxMain}, the expected approximation guarantees achieved by Algorithms~\ref{alg:distance2:sillyCyclicRandPerm}~and~\ref{alg:distance2:sillyRandPerm} are better than $e$ when any of $d$, $\frac{n}{d}$ or the largest degree in $T$ is small.
For example, if $d=2$, then both algorithms return a $2$-approximate solution in expectation.

Both algorithms are easy to de-randomize using conditional probabilities as follows.
Observe that the conditional probability of an edge being added to the solution can be easily computed under the condition that the positions of some of the nodes are already fixed.
Therefore, one can try to put each node to the first place, and choose the one that gives the highest (conditional) expectation.
Then try each of the remaining nodes at the second position and put the best one there, and so on.
The weights of the outputs of the de-randomized algorithms are clearly at least as large as the expected weight of the solutions found by the randomized algorithms.\\

In the rest of this section, an improved approximation algorithm for the non-cyclic case is presented.
Algorithm~\ref{alg:distance2:sillyRandPerm} includes an edge $st$ in the solution set if and only if $t$ has no neighbors among the $(d-1)$ nodes in front of $s$.
A more efficient approach is to run a greedy algorithm enumerating the edges incident to $t$ from left to right for each $t\in T$, in other words, the modified algorithm generates a random permutation of $S$ and executes algorithm \textsc{$T$-Greedy} as described in~\cite{madarasi2021matchings}.
Clearly, the solution returned by the modified algorithm is at least as good as the one found by Algorithm~\ref{alg:distance2:sillyRandPerm}, provided that they choose the same random permutation.
We propose the following conjecture:
\begin{conjecture}\label{conj:distance2:perm:approxConj}
  Generating a random permutation of $S$, algorithm \textsc{$T$-Greedy} finds a feasible $d$-distance $b$-matching whose expected weight is at least $\frac{1}{2}\frac{d^2+d+2}{d^2+d}>\frac{1}{2}$ times the weight of the heaviest $d$-distance $b$-matching under all permutations.
\end{conjecture}
Note that the conjecture is based on an extensive computational study.
We computer-verified the statement in all cases when $|S|\leq 1000$ and $d\leq 100$.
Enumerating all such instances directly is hopeless, but one can design a non-trivial dynamic programming algorithm for computing the exact expected value in the case $|S|=(k-1)d+1$ and $|T|=1$.
Similarly to the proof of Theorem~\ref{thm:distance2:randPerm:apxMain}, this confirms the conjecture for all problem instances with $|S|\leq 1000$ and $d\leq 100$.

Note that it is not difficult to construct an example for each $d$ in which the expected weight of the returned edges is exactly $\frac{1}{2}\frac{d^2+d+2}{d^2+d}$ times the optimum, so one cannot hope to improve the bound above.

\section{Open questions}
An FPT algorithm parameterized by $d$ was given for the $d$-distance matching problem~\cite{madarasi2021matchings}.
A straightforward generalization of this approach gives an FPT algorithm parameterized by both $d$ and $\max_{v\in V}b(v)$ for the $d$-distance $b$-matching problem.
It is not clear whether an FPT algorithm parameterized only by $d$ exists for this problem.
The $d$-distance matching problem was shown to be solvable in polynomial time when the size of $T$ is a constant~\cite{madarasi2021matchings}.
Is the problem polynomial-time solvable when the size of $T$ is taken as a parameter?

It remains open whether an optimal permutation can be found when both $d$ and the coordinates of $b$ are constants.
Improving the $e$-approximation algorithms for finding the best permutation --- and in particular proving Conjecture~\ref{conj:distance2:perm:approxConj} --- seems to be a challenging problem.

By Theorems~\ref{thm:distance2:perm:nonCyclicDMSolvable}~and~\ref{thm:distance2:perm:cyclicDMSolvable}, finding a permutation maximizing the weight of the heaviest (cyclic) $d$-distance $b$-matching can be solved when $b(s)= 1$ for all $s$, and it is NP-hard when $b(s)=2$ for all $s\in S$.
What can we say about the cases between these two extremes?
It is easy to see that the problem is solvable when $b(s)\in\{1,2\}$ and there are only a constant number of nodes for which $b$ is $2$.
A natural question is whether an FPT algorithm parameterized by the number of nodes for which $b$ is $2$ exists.

The construction given in Theorem~\ref{thm:distance2:perm:cyclicDbMNPC} seems to be easy to modify to show that the optimization version of the cyclic problem is APX-hard.
We do not know whether the non-cyclic optimization problem behaves differently in this regard.

The integrality gap of~{(LP1')} is at most $(2-\frac{1}{d})$ in the cyclic case when the size of $S$ is divisible by $(2d-1)$.
We believe that this bound holds regardless of the size of $S$, but the proof of Theorem~\ref{thm:distance2:gapCyclicDistance} does not seem to generalize to this case.

In the non-cyclic case, the integrality gap is at most $(2-\frac{2}{d})$, but this does not seem to be tight.
The exact value of the integrality gap remains unknown.

In a natural generalization of the problem, a bound $g(t,I)\in\Z_+$ is given on the number of edges induced by $I$ and $t$ for each $t\in T$ and for each interval $I$ of length $d$ in $S$.
When $g\equiv 1$, we get back the $d$-distance $b$-matching problem.
Some of the results presented in this paper also apply to other special cases of this more general setting.
For example,~(\ref{lp:distance2:dmLp}) easily extends to the more general problem by changing the right hand-side of~(\ref{lp:distance2:dmLp:eq:distanceConstr}) to $g(s_i,R_d(s_i))$, and adding $x\leq\1$.
It is not hard to see that the bounds on the integrality gap given by Theorems~\ref{thm:distance2:gapCyclicDistance}~and~\ref{thm:distance2:gapDistance} hold for any uniform $g$.

The problem has several other natural generalizations.
For example, pose distance constraints on both node classes, or drop some of the distance constraints, etc., which are subjects for further research.

Motivated by the position-based scheduling problem on a single machine~\cite{horvathKis2020positionBasedScheduling}, we introduce the \emph{position-based optimal permutation problem}, in which placing $s\in S$ at each position has an associated cost, and the goal is to find a minimum-cost permutation of $S$ under which a perfect $d$-distance matching exists.
When the cost function is uniform, the problem can be solved by Theorem~\ref{thm:distance2:perm:nonCyclicDMSolvable}.
This approach does not seem to work for other cost functions, so it is an exciting open question whether this problem is polynomial-time solvable.

\section{Acknowledgement}
The author is grateful to Sára Hanna Tóth for discussions and for finding the gadget used in the proof of Theorem~\ref{thm:distance2:doubleUnweightedApxHard}.
This research has been implemented with the support provided by the Ministry of Innovation and Technology of Hungary from the National Research, Development and Innovation Fund, financed under the ELTE TKP 2021-NKTA-62 funding scheme, by the Ministry of Innovation and Technology NRDI Office within the framework of the Artificial Intelligence National Laboratory Program, and by the Lend\"ulet Programme of the Hungarian Academy of Sciences -- grant number LP2021-1/2021.\\

% \section*{Compliance with ethical standards}
% This article does not contain any studies with human participants or animals performed by any of the authors.

\bibliographystyle{spmpsci}      % mathematics and physical sciences
\newpage
\bibliography{matchingsUnderDistanceConstrII_revised}   % name your BibTeX data base

% \tableofcontents

\end{document}